\documentclass[10pt,journal]{IEEEtran}

\usepackage{amsmath,amssymb,amsthm,mathtools,bm}
\usepackage{cite}
\usepackage[hidelinks]{hyperref}
\usepackage{microtype}
\usepackage{graphicx}
\usepackage{dsfont}
\usepackage{tikz}
\usepackage{array,tabularx,booktabs}
\usetikzlibrary{arrows.meta,positioning,calc,fit}
\allowdisplaybreaks

\newcommand{\E}{\mathbb{E}}

\newcommand{\Prb}{\mathbb{P}}

\newcolumntype{L}[1]{>{\raggedright\arraybackslash}p{#1}}
\newcolumntype{Y}{>{\raggedright\arraybackslash}X}

\newtheorem{proposition}{Proposition}
\newtheorem{lemma}{Lemma}
\newtheorem{theorem}{Theorem}
\newtheorem{remark}{Remark}

\newcommand{\MainSecNum}{Section~IV}
\newcommand{\MainThmOne}{Theorem~1}
\newcommand{\MainLemmaOne}{Lemma~1}
\newcommand{\MainThmTwo}{Theorem~2}
\newcommand{\MainPropOne}{Proposition~1}
\newcommand{\MainEqSIR}{equation~(3)}
\newcommand{\MainEqPsucc}{equation~(9)}
\newcommand{\MainEqIdeal}{equation~(13)}
\newcommand{\MainEqQlb}{equations~(14)--(15)}
\newcommand{\MainEqContraction}{equations~(16)--(17)}
\newcommand{\MainEqPstar}{equations~(20)--(21)}

\title{Wireless Broadcast Gossip for Decentralized Drone Swarms: Success Probability, Contraction, and Optimal Aloha}
\author{Ali Khalesi%
\thanks{A.~Khalesi is an \textit{Assistant Professor} at Institut Polytechnique des Sciences Avanc\'ees (IPSA) and LINCS Lab, Paris, France (ali.khalesi@ipsa.fr).}
}

\begin{document}

\maketitle

\begin{abstract}
We study a tractable baseline for average-preserving broadcast gossip in decentralized drone swarms under a quasi-static planar Poisson model and a matching-based abstraction. With slotted Aloha, Rayleigh fading, and threshold decoding, we derive: 1) a closed-form SIR success law; 2) a mean-square contraction bound that separates ideal mixing from wireless successful updates via a conservative lower bound; and 3) a closed-form proxy access rule with interpretable density scaling. Explicit-interference simulations, together with robustness checks for receiver selection, noise, fading, and spatial regularity, confirm a stable intermediate operating region for the Aloha probability.
\end{abstract}


% =========================
\section{Introduction}
\label{sec:intro}
% =========================

Drone swarms increasingly perform coordination tasks without a central
controller: maintaining formations, aligning headings, agreeing on shared
reference frames, and fusing local measurements to build a common estimate.
In many practical settings, the only scalable communication primitive is
short-range \emph{broadcast}: each drone periodically sends a small message
that nearby drones may receive. Because these interactions are local and do not
require routing or a leader, \emph{consensus} methods based on pairwise
averaging, often referred to as \emph{gossip}
algorithms,\footnote{A gossip algorithm is an iterative scheme in which nodes
repeatedly combine their values using only local exchanges; in its simplest
form, two nodes replace their values by their average, which preserves the
global average and gradually removes disagreements.}
are a natural fit for swarm coordination
\cite{BoydGossip2006,OlfatiSaber2007,RenBeard2005}.

A key difficulty is that swarm communication does not occur over reliable
links. Broadcasts share a common wireless channel, so simultaneous
transmissions interfere with each other. As a result, whether a packet is
decoded is random and depends on how many other drones transmit at the same
time and how close they are. Hence, the speed of a gossip loop is determined
not only by \emph{who could talk to whom} through the neighborhood graph, but
also by \emph{how often messages are successfully received} under interference.
Capturing this interaction between wireless reliability and averaging dynamics
is the main goal of this letter. While we use a standard threshold-based
decoding model, the focus is not on detailed physical-layer design. Rather, we
seek a compact analytical relation between medium-access parameters and the
speed of average-preserving consensus.

The literature most closely related to this paper spans several directions.
Classical randomized gossip and consensus works establish convergence and mixing
properties for local averaging over static or time-varying graphs
\cite{BoydGossip2006,OlfatiSaber2007,RenBeard2005}. Consensus over random or
unreliable networks has also been studied through random graph evolution, link
failures, and noisy communication models
\cite{HatanoMesbahi2005,TahbazSalehiJadbabaie2008,TahbazSalehiJadbabaie2010,KarMouraRandomLinks2008,KarMouraImperfect2009,PorfiriStilwell2007,FagnaniZampieri2008}.
At a broader level, surveys such as \cite{DimakisGossip2010} emphasize that
gossip-type algorithms are especially attractive in wireless settings because
they avoid specialized routing and remain compatible with unreliable links and
distributed processing constraints.

A second line of work is closer in spirit to wireless broadcast operation.
Broadcast gossip algorithms have been developed precisely to exploit the
broadcast nature of wireless exchanges rather than only pairwise routed links
\cite{AysalBroadcastConsensus2009,FranceschelliBroadcast2011}. Related wireless
consensus and synchronization protocols have also been studied in more
communication-aware settings, for example under multipath fading or in
practical wireless sensor protocols
\cite{ScutariBarbarossa2008,SchenatoFiorentin2011}. These works show that the
wireless medium can materially affect convergence behavior, but they typically
do not lead to a compact interference-driven MAC tuning rule of the type
considered here.

Our viewpoint sits at the interface between these literatures. On one hand,
consensus over random or unreliable networks is often analyzed through abstract
graph processes, link failures, or noisy exchanges without an explicit
medium-access and interference model. On the other hand, stochastic-geometry
analyses of wireless networks provide explicit reception laws under random
access, but typically for coverage, throughput, or connectivity metrics rather
than disagreement contraction. The present contribution combines these
viewpoints: it starts from an explicit wireless interference model, imposes an
average-preserving gossip abstraction, and derives a compact MAC-level
operating rule for the resulting interference--mixing tradeoff.

A standard way to obtain tractable yet informative reliability laws is
\emph{stochastic geometry},\footnote{Stochastic geometry models node locations
as random points in space. This makes it possible to compute network-wide
quantities, such as the probability of successful reception, by averaging over
the random geometry.}
and, in particular, the \emph{Poisson point process (PPP)}
baseline.\footnote{A PPP is a random point pattern in which the number of
points in disjoint regions are independent and Poisson distributed. Its single
parameter is the spatial density $\lambda$ (points per unit area). It is widely
used as a tractable baseline for irregular deployments.}
Under Rayleigh fading\footnote{Rayleigh fading models rapid small-scale channel
fluctuations; under this model the received power gain is exponential.}
and power-law path loss, PPP-based models yield closed-form expressions for the
probability that a transmission meets a target
\emph{signal-to-interference ratio (SIR)}.\footnote{SIR is the received signal
power divided by the sum of interference powers from other transmitters. In
this paper, the threshold $\theta$ is interpreted as an effective PHY operating
requirement that encapsulates modulation, coding, and receiver operating
point.}
These tools were developed extensively in the Poisson-network literature; see,
e.g., the monographs \cite{BaccelliBlasz2009v1,BaccelliBlasz2010v2} and
studies of Aloha-type random access in Poisson networks
\cite{BaccelliTIT2006Aloha,BaccelliJSAC2009Aloha,BaccelliSinghWiOpt2013}. The
PPP assumption should therefore be interpreted here as a tractable
irregular-deployment baseline rather than as a claim that real swarms literally
form a Poisson pattern. More structured formations, such as hard-core or
perturbed lattice layouts, would modify the interference statistics and hence
the constants in the resulting design rule, but the PPP baseline cleanly
exposes the density--interference tradeoff that we wish to analyze.

We adopt a quasi-static snapshot model appropriate for short control loops.
Drone locations are modeled as a planar PPP of intensity $\lambda$, and time is
slotted. In each slot, each drone transmits with probability $p$ using
\emph{slotted Aloha}.\footnote{Slotted Aloha is a random-access method where
each node independently decides to transmit in each time slot with probability
$p$.}
The geometry is treated as approximately constant over the short time window
used to characterize one-step and short-horizon contraction; this approximation
is most appropriate when the typical displacement per slot satisfies
$v\Delta t \ll R$, where $R$ is the communication radius. For example, when
\(R\) is on the order of tens of meters and the slot duration is in the
millisecond range, the resulting per-slot displacement can remain small enough
for the quasi-static approximation to be meaningful over short control windows.
Each listening drone considers neighbors within range $R$ and attempts to
decode one candidate broadcast. Successful decoding is determined by an
interference-limited SIR rule, which is most relevant in dense short-range
regimes where concurrent transmissions dominate receiver noise; in weaker
interference regimes, a SINR model would be more appropriate. Likewise,
Rayleigh fading is adopted as a tractable baseline and can be viewed as a
conservative proxy relative to more LoS-dominated UAV channels.

Because successful broadcast reception is inherently one-way, we introduce a
\emph{matching-based analytical abstraction} to obtain an average-preserving
update rule. Specifically, each slot is mapped to a set of symmetric pairwise
exchanges forming a matching, so that each executed exchange becomes a pairwise
averaging step and the resulting update matrix is doubly stochastic. This
abstraction deliberately underestimates the full richness of broadcast
reception, since a transmission may in practice be decodable by multiple
receivers, but it isolates the interaction between wireless success events and
average-preserving gossip in a mathematically clean way. It should therefore be
read as an analysis layer rather than as a complete protocol specification:
explicit handshake, contention resolution, or acknowledgment overhead is not
modeled here and would effectively reduce the realized successful-update rate.
Similarly, richer receiver behaviors such as multi-receiver capture or
strongest-signal selection may improve spreading in practice, but would move
the state evolution away from the present pairwise average-preserving baseline.

\textbf{Design question.}
Given the swarm density $\lambda$, neighborhood radius $R$, and decoding
threshold $\theta$, what transmit probability $p$ yields the \emph{fastest}
convergence? When $p$ is too small, few updates occur because the medium is
mostly idle. When $p$ is too large, many broadcasts collide through
interference and reception fails. The main message of this letter is that this
tradeoff can be converted into a simple and interpretable access guideline.

\textbf{Contributions.}
Under the interference-limited PPP+Aloha baseline model of
Section~\ref{sec:model}, we derive:
\begin{itemize}
\item \textbf{Closed-form reception probability:} an explicit expression for
the SIR success probability as a function of
$(\lambda,p,r,\theta,\alpha)$ (Theorem~\ref{thm:psucc}).
\item \textbf{Convergence with wireless thinning:} a mean-square contraction
bound in which the consensus rate factorizes into an \emph{ideal} mixing term
and an explicit \emph{wireless} success term through a conservative per-slot
successful-update lower bound (Theorem~\ref{thm:contraction_clean}).
\item \textbf{Proxy-based optimal access:} a closed-form access rule
$p^\star$ for a tractable availability--reliability proxy, revealing the
scaling $p^\star=\Theta(1/\lambda)$ for fixed $(R,\theta,\alpha)$
(Proposition~\ref{prop:pstar_clean}).
\end{itemize}

\textbf{Interpretation and scope.}
The present letter deliberately focuses on a tractable baseline consisting of a
quasi-static PPP geometry, an interference-limited threshold model, Rayleigh
fading, and a matching-based average-preserving abstraction. It does not
attempt to provide a full protocol-level treatment of control signaling,
mobility-induced temporal correlation, LoS-specific propagation, or structured
swarm formations. Accordingly, the resulting guideline should be interpreted as
a density-adaptive MAC heuristic under the adopted model assumptions, rather
than as a universal exact optimizer for all UAV communication settings.

\textbf{Practical implication for drone swarms.}
Within this baseline, the resulting message is simple: denser deployments
should generally use a lower randomized broadcast duty cycle to avoid
interference-driven reception failures. Numerical evidence supporting this rule
is summarized in Section~\ref{sec:sim}, while full experimental details are
deferred to the supplementary material.

\textbf{Organization.}
Section~\ref{sec:model} presents the system model and the average-preserving
broadcast-gossip abstraction. Section~\ref{sec:results} states the main
analytical results: the PPP+Aloha SIR success probability, the wireless-thinned
contraction bound, and the proxy-based access rule \(p^\star\).
Section~\ref{sec:sim} gives a compact numerical-validation summary.
Section~\ref{sec:conclusion} concludes the letter and outlines directions for
future work.

\section{System Model}
\label{sec:model}

\textbf{Geometry and state model:}
Let \(\Phi\subset\mathbb{R}^2\) be a stationary planar PPP of intensity
\(\lambda\). We study the swarm inside a large observation window
\(\mathcal{W}\subset\mathbb{R}^2\) and denote the drone set by
\(\mathcal{V}\triangleq \Phi\cap\mathcal{W}\), with
\(n\triangleq|\mathcal{V}|\). Condition on \(n\ge 2\), and index drones by
\(i\in\{1,\dots,n\}\) with locations \(X_i\in\mathcal{W}\). Each drone
maintains a scalar state \(x_i(t)\in\mathbb{R}\), for example a heading
component or formation offset. Let \(x(t)\in\mathbb{R}^n\) collect all states
and define the target average
\begin{equation}
\bar x \triangleq \frac{1}{n}\sum_{i=1}^n x_i(0),
\qquad
\text{goal: } x_i(t)\to \bar x \text{ for all } i.
\end{equation}
The PPP assumption is used here as a tractable baseline
for irregular deployments rather than as a literal model
of formation-controlled swarms. Likewise, the large-window
formulation neglects finite-area boundary effects.

\textbf{Broadcast neighborhood and slotted access:}
Time is slotted. In each slot, drone \(i\) transmits independently with
probability \(p\) according to slotted Aloha. Let
\(\mathcal{T}(t)\subseteq\mathcal{V}\) denote the set of transmitters at slot
\(t\). Under the PPP baseline, \(\mathcal{T}(t)\) is a \(p\)-thinning of
\(\mathcal{V}\). Drones are half-duplex, so a transmitting node cannot receive
in the same slot. Each drone \(i\) considers neighbors within range \(R\):
\begin{equation}
\mathcal{N}_i \triangleq
\{j\in\{1,\dots,n\}\setminus\{i\}:\|X_j-X_i\|\le R\}.
\end{equation}
The analysis adopts a quasi-static snapshot viewpoint: geometry is treated as
approximately constant over the short time window used to characterize one-step
and short-horizon contraction. This is most appropriate when the typical
displacement per slot satisfies
\[
v\Delta t \ll R,
\]
where \(v\) is a representative speed and \(\Delta t\) is the slot duration.
For example, if \(R\) is on the order of tens of meters and \(\Delta t\) is in
the millisecond range, then the per-slot displacement can remain small enough
for the geometry to be effectively frozen over short control windows.
Accordingly, the present model is intended for short-horizon contraction rather
than long-horizon mobility-dominated evolution. For a horizon of \(T\) slots, a
conservative order-of-magnitude validity condition is
\[
T v \Delta t \ll R,
\]
or, more generally, that the neighbor set and the dominant
interferers do not change drastically over the contraction
window. Longer horizons may therefore be interpreted as a
sequence of piecewise quasi-static blocks, to which the present
analysis applies locally rather than globally over the entire run.

\textbf{Wireless channel and threshold decoding:}
Path-loss is \(\ell(r)=r^{-\alpha}\) with \(\alpha>2\), and fading is
Rayleigh, so the power gain \(h_{ji}(t)\sim\mathrm{Exp}(1)\) is i.i.d.\ across
links and slots. Conditioned on the active transmitters \(\mathcal{T}(t)\),
the SIR at receiver \(i\) for transmitter \(j\in\mathcal{T}(t)\) is
\begin{equation}
\label{eq:sir_def}
\mathrm{SIR}_{j\to i}(t)
\triangleq
\frac{h_{ji}(t)\,\|X_j-X_i\|^{-\alpha}}
{\sum_{k\in\mathcal{T}(t)\setminus\{j\}} h_{ki}(t)\,\|X_k-X_i\|^{-\alpha}}.
\end{equation}
A packet is declared decodable if \(\mathrm{SIR}_{j\to i}(t)>\theta\) for a
fixed threshold \(\theta>0\). This interference-limited formulation is
meant for dense short-range regimes in which concurrent transmissions
dominate the noise term over the time scale of interest. In weaker-interference
regimes, a SINR model would be more appropriate. The threshold \(\theta\)
is interpreted as an effective PHY operating requirement that encapsulates
modulation, coding, and receiver operating point. Rayleigh fading is adopted
as a tractable baseline and may be viewed as a conservative proxy relative
to more LoS-dominated UAV channels. A schematic illustration of the model
and an illustrative PHY/MAC parameter table are provided in the
supplementary material.

\textbf{Receiver and update abstraction:}
Because transmissions are broadcast, a listening drone may in principle decode
one or more nearby transmitters. To obtain an analytically clean
average-preserving iteration, we adopt a conservative receiver/update
abstraction. Specifically, each listening drone
\(i\notin\mathcal{T}(t)\) selects one candidate transmitting neighbor in
\(\mathcal{N}_i\cap\mathcal{T}(t)\), for example uniformly when multiple
candidates are present, and attempts to decode only that selected link. This
choice underestimates the full richness of broadcast reception, but it keeps
the state-update rule pairwise and interpretable, and it preserves a
doubly-stochastic average-preserving structure once combined with the matching
step below. More expressive receiver models, such as strongest-signal
selection, multi-receiver decoding, or capture-based rules, may improve
information spreading in practice, but they would also modify the present
pairwise baseline and are therefore left outside the scope of the analysis.

To preserve the global average, successful receptions are then mapped to a set
of executed exchanges forming a matching \(\mathcal{M}(t)\) on the node set
\(\{1,\dots,n\}\). Each matched pair \((i,j)\in\mathcal{M}(t)\) performs
pairwise averaging:
\begin{equation}
x_i(t{+}1)=x_j(t{+}1)=\tfrac12\bigl(x_i(t)+x_j(t)\bigr),
\end{equation}
and all unmatched nodes keep their states. Equivalently,
\begin{equation}
x(t{+}1)=W(t)\,x(t),
\end{equation}
where \(W(t)\) is random, symmetric, and doubly stochastic. A convenient
representation is
\begin{equation}
\label{eq:W_matching}
W(t)=I-\frac12\sum_{(i,j)\in\mathcal{M}(t)}(e_i-e_j)(e_i-e_j)^\top,
\end{equation}
with \(\{e_i\}\) the standard basis of \(\mathbb{R}^n\).

\textbf{Performance metric:}
We track the standard disagreement energy
\begin{equation}
V(t)\triangleq \|x(t)-\bar x\,\mathbf{1}\|_2^2,
\end{equation}
which is non-increasing under ideal pairwise averaging and is used to state
the contraction bounds.

% =========================
\section{Main Results}
\label{sec:results}
% =========================

We present three results aligned with the baseline model of
Section~\ref{sec:model}: (i) a closed-form SIR success probability under
PPP+Aloha; (ii) a contraction bound for the average-preserving
broadcast-gossip abstraction under wireless thinning; and (iii) a proxy-based
access rule \(p^\star\). Detailed proofs and additional interpretive remarks
are deferred to the supplementary material.

\subsection{Success probability under PPP+Aloha}

Fix a receiver at the origin under the Palm distribution and consider a desired
transmitter at distance \(r>0\).

\begin{theorem}[SIR success probability]\label{thm:psucc}
Under the PPP+Aloha model with intensity \(\lambda\), Aloha probability \(p\),
Rayleigh fading, and path-loss \(\ell(r)=r^{-\alpha}\) with \(\alpha>2\), the
success probability of a transmission over distance \(r\) is
\begin{align}
\label{eq:psucc}
p_{\mathrm{succ}}(r;p)
\;&\triangleq\;
\Prb\!\big(\mathrm{SIR}>\theta\big)
\\
&=
\exp\!\Big(-\lambda p\,\pi r^2\,\theta^{2/\alpha}\,C(\alpha)\Big),
\end{align}
where
\begin{equation}
\label{eq:Calpha}
C(\alpha)=\Gamma\!\Big(1+\tfrac{2}{\alpha}\Big)\Gamma\!\Big(1-\tfrac{2}{\alpha}\Big)
=\frac{2\pi}{\alpha}\csc\!\Big(\frac{2\pi}{\alpha}\Big).
\end{equation}
\end{theorem}

\begin{IEEEproof}[Proof of Theorem~\ref{thm:psucc}]
See Supplementary Appendix~B.
\end{IEEEproof}

\subsection{Wireless-thinned contraction bound}

We next relate consensus speed to two distinct ingredients:
(a) an \emph{ideal-mixing} factor associated with the average-preserving update
abstraction, and
(b) a \emph{wireless successful-update} factor induced by the PPP+Aloha link
model.

\paragraph{Ideal mixing.}
Let \(W(t)\) be the random symmetric, doubly stochastic update matrix induced
by the matching-based abstraction of Section~\ref{sec:model}. Consider the
idealized case in which every scheduled exchange succeeds. Define the
disagreement projector
\(\Pi\triangleq I-\frac{1}{n}\mathbf{1}\mathbf{1}^\top\) and
\begin{equation}
V(t)\triangleq \|x(t)-\bar x\,\mathbf{1}\|_2^2=x(t)^\top \Pi x(t).
\end{equation}

\begin{lemma}[Ideal one-step contraction]\label{lem:ideal_gap_min}
Let \(x(t{+}1)=W(t)x(t)\) with \(W(t)\) symmetric and doubly stochastic,
i.i.d.\ over \(t\). Define
\begin{equation}
\label{eq:rho_def}
\rho \triangleq \lambda_{\max}\!\Big(\E\!\big[W(t)^\top \Pi W(t)\big]\Big)\in[0,1),
\qquad
\gamma \triangleq 1-\rho.
\end{equation}
Then, for the ideal always-successful process,
\begin{equation}
\label{eq:ideal_step_min}
\E[V(t{+}1)\mid x(t)] \le \rho\,V(t)=(1-\gamma)V(t).
\end{equation}
\end{lemma}

\begin{IEEEproof}[Proof of Lemma~\ref{lem:ideal_gap_min}]
See Supplementary Appendix~C.
\end{IEEEproof}

\paragraph{A conservative per-slot successful-update lower bound.}
A typical node updates in slot \(t\) if (i) it listens, (ii) at least one
neighbor within range \(R\) transmits, and (iii) at least one such
transmission is decodable. A conservative lower bound is obtained by requiring
that there exists at least one transmitter in the ball \(B(0,R)\) and that a
transmitter at the worst-case distance \(R\) succeeds. This yields
\begin{align}
\label{eq:q_lb_def_clean}
q_{\mathrm{lb}}(p)
&\triangleq
(1-p)\Big(1-e^{-\lambda p \pi R^2}\Big)e^{-\lambda p K(R)},
\\
K(R)&=\pi R^2\theta^{2/\alpha}C(\alpha).
\end{align}

\begin{theorem}[Consensus contraction under wireless thinning]\label{thm:contraction_clean}
Let \(V(t)\triangleq \|x(t)-\bar x\,\mathbf{1}\|_2^2\) be the disagreement
energy. Under the baseline model of Section~\ref{sec:model}, we have
\begin{equation}
\label{eq:contraction_clean}
\E[V(t)] \le \big(1-\gamma\,q_{\mathrm{lb}}(p)\big)^t\,V(0),
\end{equation}
where \(\gamma\) is defined in \eqref{eq:rho_def} and \(q_{\mathrm{lb}}(p)\)
in \eqref{eq:q_lb_def_clean}. Consequently, the \(\varepsilon\)-consensus time
satisfies
\begin{equation}
\label{eq:Teps_clean}
T_\varepsilon
\le
\frac{\log\!\big(V(0)/\varepsilon\big)}{-\log\!\big(1-\gamma q_{\mathrm{lb}}(p)\big)}
\le
\frac{1}{\gamma\,q_{\mathrm{lb}}(p)}\log\!\Big(\frac{V(0)}{\varepsilon}\Big).
\end{equation}
\end{theorem}

\begin{IEEEproof}[Proof of Theorem~\ref{thm:contraction_clean}]
See Supplementary Appendix~C.
\end{IEEEproof}

\subsection{Optimal Aloha probability}

A tractable design rule is to choose \(p\) that maximizes a successful-update
proxy derived from the lower bound above.

\begin{proposition}[Closed-form \(p^\star\) for a proxy availability--reliability rule]\label{prop:pstar_clean}
Let
\begin{equation}
a \triangleq \lambda\pi R^2,\qquad
b \triangleq \lambda K(R)=\lambda\pi R^2\theta^{2/\alpha}C(\alpha).
\end{equation}
Consider the simplified proxy
\begin{equation}
\tilde q(p)\triangleq \big(1-e^{-ap}\big)e^{-bp},\qquad p\in[0,1],
\end{equation}
obtained by dropping the half-duplex factor \(1-p\). Since this omission is
mild when \(p\) is small to moderate, \(\tilde q(p)\) should primarily be
viewed as an operating-region proxy and scaling rule rather than as a fully
faithful model of the true successful-update probability. Its maximizer is
\begin{equation}
\label{eq:pstar_closed}
p^\star = \min\!\Big(1,\ \frac{1}{a}\log\!\Big(\frac{a+b}{b}\Big)\Big).
\end{equation}
Moreover, in the dense-neighborhood regime where \(1-e^{-ap}\approx 1\), the
maximizer further simplifies to
\begin{equation}
\label{eq:pstar_simple_clean}
p^\star \approx \min\!\Big(1,\ \frac{1}{b}\Big)
=\min\!\Big(1,\ \frac{1}{\lambda\pi R^2\theta^{2/\alpha}C(\alpha)}\Big).
\end{equation}
\end{proposition}

\begin{IEEEproof}[Proof of Proposition~\ref{prop:pstar_clean}]
See Supplementary Appendix~D.
\end{IEEEproof}

\section{Numerical Validation}
\label{sec:sim}

The analytical predictions were validated through explicit-interference
simulations that do not sample packet outcomes from the closed-form PPP success
law. Across the considered experiments, the disagreement metric exhibits a
clear non-monotone dependence on the Aloha access probability \(p\), confirming
a stable intermediate operating region. Additional studies also show that the
worst-case analytical surrogates are conservative, while refined empirical
distance averaging improves numerical prediction. The same qualitative design
message remains robust under receiver-rule changes, additive noise, alternative
fading laws, and more regular spatial layouts. Full simulation details,
figures, proxy-accuracy comparisons, lower-bound tightness studies, and
robustness experiments are provided in the supplementary material.

\section{Conclusion}
\label{sec:conclusion}

This letter studied a tractable baseline question for decentralized drone
swarms: how random broadcast gossip contracts disagreement when the wireless
medium is shared, interference-limited, and accessed through slotted Aloha.
Under a Poisson spatial model with threshold-based decoding, we obtained three
main analytical results. First, Theorem~\ref{thm:psucc} gives a closed-form
reception law that makes the dependence on density, access probability, link
distance, and the SIR threshold explicit. Second,
Theorem~\ref{thm:contraction_clean} yields a mean-square contraction bound in
which the convergence rate separates into an \emph{ideal-mixing} component,
capturing the strength of the average-preserving update abstraction, and a
\emph{wireless successful-update} component, capturing how often useful
exchanges occur under interference. Third,
Proposition~\ref{prop:pstar_clean} provides an explicit access rule \(p^\star\)
for an availability--reliability proxy and reveals the interpretable scaling
\(p^\star=\Theta(1/\lambda)\) for fixed \((R,\theta,\alpha)\).

The numerical results support this design message while also clarifying its
scope. Explicit-interference simulations confirm that the disagreement curve is
unimodal in \(p\), so the best operating point lies in an intermediate range.
Further proxy-accuracy, tightness, and robustness results are reported in the
supplementary material.

Accordingly, the present results should be interpreted as a tractable baseline
benchmark and a density-adaptive MAC guideline under the adopted abstraction,
not as a universal exact optimizer for all UAV deployments. Several
simplifying assumptions were made deliberately to keep the analysis
transparent. The geometry is treated as quasi-static over the contraction
horizon, decoding is described through a threshold model, Rayleigh fading is
used as a tractable baseline, and successful broadcast receptions are mapped to
pairwise matchings in order to preserve the global average. These choices make
the problem analytically clean, but they also leave out protocol overhead,
richer multi-receiver behavior, strongly structured formations, and
long-horizon mobility effects. The contribution of the letter is therefore less
a complete protocol specification than a compact analytical framework that
exposes the main interference--mixing tradeoff and turns it into an
interpretable access-design rule.

\paragraph*{Future work}
A natural next step is to relax these baseline assumptions in a controlled way.
One important direction is to move beyond the static snapshot and study
time-varying geometries induced by mobility, including heterogeneous speeds and
direction changes, in order to understand how temporal diversity modifies both
interference and mixing. It is also important to incorporate more structured
spatial models, explicit SINR parameterizations with concrete PHY quantities,
and richer air-to-air fading laws. On the algorithmic side, it would be useful
to replace the current matching abstraction by a more protocol-aware broadcast
update model that accounts for coordination, acknowledgments, or scheduling
overhead while still preserving average consensus. A further useful direction
is to characterize more sharply the gap between the conservative lower-bound
successful-update term and the realized node-level update probability under
explicit interference. Such extensions would help bridge the gap between the
present tractable baseline and more realistic swarm communication settings.

\bibliographystyle{IEEEtran}

\clearpage

\begin{center}
{\LARGE Supplementary Material for}\\[0.25em]
{\LARGE Wireless Broadcast Gossip for Decentralized Drone Swarms: Success Probability, Contraction, and Optimal Aloha}\\[0.7em]
{\large Ali Khalesi}
\end{center}

\noindent\textit{Reference convention:} Throughout this supplement, references such as
\MainThmOne, \MainThmTwo, \MainPropOne, and \MainEqPsucc{} refer to the numbering in the
main manuscript. Figures and tables appearing only in this supplement are labeled
with the prefix ``S'' to avoid ambiguity.

\appendices

\section{Additional Model Clarifications}

\begin{table*}[!t]
\caption{Illustrative PHY/MAC parameterization used only to interpret the
normalized threshold model and the interference-limited regime. These values
are not used in the theorem statements.}
\label{tab:phy_params_supp}
\centering
\small
\setlength{\tabcolsep}{4pt}
\renewcommand{\arraystretch}{1.08}
\begin{tabularx}{\textwidth}{@{}L{0.24\textwidth}L{0.20\textwidth}Y@{}}
\toprule
\textbf{Parameter} & \textbf{Example value} & \textbf{Role} \\
\midrule
Carrier frequency \(f_c\) & \(2.4\) GHz & short-range unlicensed link \\
Bandwidth \(B\) & \(5\) MHz & narrowband control channel \\
Transmit power \(P_{\mathrm{tx}}\) & \(20\) dBm & representative short-range power \\
Noise figure \(F\) & \(6\) dB & receiver quality assumption \\
Thermal noise density \(N_0\) & \(-174\) dBm/Hz & standard reference value \\
Noise power \(N_0B+F\) & \(\approx -101\) dBm & for \(B=5\) MHz and \(F=6\) dB \\
Decoding threshold \(\theta\) & \(0\) dB (\(\theta=1\)) & effective PHY operating point \\
\bottomrule
\end{tabularx}
\end{table*}

\textbf{Representative PHY interpretation:}
Although the baseline analysis is written in normalized form, the threshold
model can be interpreted relative to a conventional narrowband wireless link
budget. For transmit power \(P_{\mathrm{tx}}\), bandwidth \(B\), receiver noise
figure \(F\), and thermal spectral density \(N_0\), the additive-noise term is
of order \(N_0BF\), while \(\theta\) represents the effective decoding
requirement associated with a chosen modulation/coding operating point.
Representative short-range settings may involve bandwidths on the order of
\(1\) to \(10\) MHz, transmit powers in the tens to hundreds of milliwatts,
and standard receiver noise figures; under dense co-channel reuse, the
aggregate interference can then dominate the thermal-noise term over the
operating region of interest.

As an order-of-magnitude illustration, at \(f_c=2.4\) GHz and distance
\(r=50\) m, the free-space path loss is about \(74\) dB, so a transmitter with
\(P_{\mathrm{tx}}=20\) dBm yields a nominal received power near \(-54\) dBm
before fading, antenna effects, and additional implementation losses. Compared
with the noise level in Table~\ref{tab:phy_params_supp}, this leaves several
tens of decibels of short-range margin, so once multiple co-channel
transmitters are active the operating regime is plausibly interference-limited.
The purpose of this table is not to redefine the analysis in link-budget form,
but to anchor the normalized SIR/SINR abstraction and the robustness experiment
summarized in \MainSecNum{} of the main manuscript in a representative
physical regime.

\begin{figure*}[!t]
\centering
\resizebox{\textwidth}{!}{%
\begin{tikzpicture}[
  font=\small,
  arr/.style={-{Latex[length=2.6mm]}, thick},
  darr/.style={-{Latex[length=2.4mm]}, thick, dashed, black!55},
  ball/.style={black!35, thick},
  node/.style={circle, draw, thick, inner sep=1.2pt, minimum size=15pt},
  tx/.style={node, fill=black!12},
  rx/.style={node, fill=white},
  inf/.style={node, draw=black!40},
  box/.style={draw, rounded corners, thick, fill=white, inner sep=7pt, align=left},
  lab/.style={fill=white, inner sep=2pt},
]
\node[rx] (i) at (0,0) {$i$};
\node[tx] (j) at (1.55,0.85) {$j$};

\draw[ball] (i) circle (2.05);
\node[lab] at (-0.95,2.05) {$B(i,R)$};

\node[inf] (k1) at (-1.55, 1.35) {$k$};
\node[inf] (k2) at (-1.35,-1.20) {$k$};
\node[inf] (k3) at ( 0.30, 2.55) {$k$};
\node[inf] (k4) at ( 2.55, 2.15) {$k$};
\node[inf] (k5) at ( 2.95,-0.35) {$k$};

\draw[arr]  (j) -- (i);
\draw[darr] (k1) -- (i);
\draw[darr] (k2) -- (i);
\draw[darr] (k3) -- (i);
\draw[darr] (k4) -- (i);
\draw[darr] (k5) -- (i);

\node[box, anchor=west, minimum width=8.3cm] (selbox) at (7.60,2.65) {%
\textbf{Selected link}\\[-0.25em]
$r=\|X_j-X_i\|$};

\node[box, anchor=west, minimum width=8.3cm] (phy) at (7.60,1.25) {%
\textbf{Wireless link model (baseline)}\\[-0.25em]
$\ell(r)=r^{-\alpha}$, $\alpha>2$;\; Rayleigh $h\sim\mathrm{Exp}(1)$.\\
Decode if $\mathrm{SIR}_{j\to i}(t)>\theta$ (see \MainEqSIR{} of the main manuscript).};

\node[box, anchor=west, minimum width=8.3cm] (mac) at (7.60,-0.25) {%
\textbf{MAC + receiver abstraction (slot $t$)}\\[-0.25em]
Transmit w.p.\ $p$ (Aloha), listen w.p.\ $1-p$ (half-duplex).\\
Listener $i$ selects one TX neighbor in $B(i,R)$ (e.g., uniformly).};

\node[box, anchor=west, minimum width=8.3cm] (upd) at (7.60,-1.75) {%
\textbf{Average-preserving update abstraction}\\[-0.25em]
Successful receptions are mapped to a matching $\mathcal{M}(t)$.\\
For $(i,j)\in\mathcal{M}(t)$:\;
$x_i(t{+}1)=x_j(t{+}1)=\tfrac12(x_i(t)+x_j(t))$.};

\coordinate (midlink) at ($(i)!0.55!(j)$);
\coordinate (perp)    at ($(midlink)!0.22!90:(j)$);
\draw[arr] (midlink) -- (perp) |- (selbox.west);

\coordinate (tap1) at (2.75, 1.05);
\coordinate (tap2) at (2.75, 0.10);
\coordinate (tap3) at (2.75,-0.95);

\draw[arr] (tap1) to[out=0,in=180] (phy.west);
\draw[arr] (tap2) to[out=0,in=180] (mac.west);
\draw[arr] (tap3) to[out=0,in=180] (upd.west);

\node[box, anchor=west, minimum width=16.1cm] (leg) at (-1.95,-3.15) {%
\textbf{Symbols:}\;
$\Phi$ PPP intensity $\lambda$;\;
$\mathcal{T}(t)$ active TX ($p$-thinning);\;
$R$ neighbor radius;\;
$\theta$ SIR threshold.\quad
\raisebox{0.15em}{\tikz\draw[thick,fill=black!12] (0,0) circle (2.2pt);} TX\;
\raisebox{0.15em}{\tikz\draw[thick] (0,0) circle (2.2pt);} RX\;
\raisebox{0.15em}{\tikz\draw[thick,dashed,-{Latex[length=2.1mm]}] (0,0) -- (0.6,0);} interference
};
\end{tikzpicture}%
}
\caption{System model for one slot of average-preserving wireless broadcast gossip.
Drone locations are modeled by a planar PPP of intensity $\lambda$.
In each slot, each drone transmits with probability $p$ and otherwise listens.
A typical listener $i$ considers neighbors within range $R$ and selects one
transmitting neighbor $j$.
Decoding succeeds if $\mathrm{SIR}_{j\to i}(t)>\theta$ under Rayleigh fading
and path-loss exponent $\alpha>2$.
To obtain an average-preserving iteration, successful receptions are
\emph{abstracted} as a matching $\mathcal{M}(t)$ so that each node
participates in at most one symmetric averaging step per slot, yielding a
doubly stochastic update $x(t{+}1)=W(t)x(t)$.}
\label{fig:sysmodel_supp}
\end{figure*}

\begin{remark}[Interpretation of the matching abstraction]
The matching step is introduced to preserve average consensus
and enable clean spectral contraction analysis. It should
not be interpreted as a claim that practical UAV broadcast
systems naturally realize one matched exchange per slot
without coordination cost. Any handshake, acknowledgment,
contention-resolution, or scheduling overhead would effectively
reduce the realized successful-update rate and can therefore
be viewed as an additional thinning factor. Equivalently, one
may interpret the present baseline as capturing the idealized
wireless success process up to a multiplicative implementation
factor \(\eta_{\mathrm{ctl}}\in(0,1]\), where smaller
\(\eta_{\mathrm{ctl}}\) corresponds to larger coordination
overhead. In this interpretation, protocol overhead primarily
rescales the realized successful-update term rather than changing
the qualitative existence of an intermediate optimal access region.
\end{remark}

\begin{remark}[Simple control-overhead sensitivity]
A minimal way to parameterize handshake/ACK/backoff cost is
to assume that only a fraction \(\eta_{\mathrm{ctl}}(p)\in(0,1]\)
of wireless-feasible matched proposals are ultimately executed
after control signaling. For example, a constant control fraction
\(\beta\in[0,1)\) gives the crude surrogate
\(\eta_{\mathrm{ctl}}=1-\beta\), while an average of
\(c_{\mathrm{ctl}}\) short control mini-slots per executed exchange
suggests \(\eta_{\mathrm{ctl}}\approx (1+c_{\mathrm{ctl}})^{-1}\).
Under this interpretation, the contraction law of \MainThmTwo{}
is heuristically replaced by
\begin{equation}
\label{eq:eta_ctl_surrogate}
\mathbb{E}[V(t)]
\le
\bigl(1-\gamma\,\eta_{\mathrm{ctl}}(p)\,q_{\mathrm{lb}}(p)\bigr)^t V(0).
\end{equation}
If \(\eta_{\mathrm{ctl}}(p)\) is approximately constant in \(p\),
protocol overhead mainly rescales the convergence time without
changing the proxy-optimal access probability. If
\(\eta_{\mathrm{ctl}}(p)\) decreases with \(p\), the preferred
operating region shifts mildly toward smaller access probabilities.
\end{remark}

\section{Proof of \MainThmOne}

Condition on the intended link length \(r\) and on the PPP of
interferers. Under Rayleigh fading,
\begin{equation}
\Prb(\mathrm{SIR}>\theta \mid I)
=
\E\!\left[e^{-\theta r^\alpha I}\mid I\right]
=
\mathcal{L}_I(\theta r^\alpha),
\end{equation}
where \(I=\sum_{k\in\Phi_t\setminus\{0\}} h_k \|X_k\|^{-\alpha}\) is the
aggregate interference and \(\Phi_t\) is the PPP of active transmitters, i.e.,
the \(p\)-thinning of \(\Phi\), hence a PPP with intensity \(\lambda p\).
Using the probability generating functional of the PPP,
\begin{align}
\mathcal{L}_I(s)
&=
\E\!\left[\prod_{X\in\Phi_t} \E_h\!\left(e^{-s h \|X\|^{-\alpha}}\right)\right]\nonumber\\
&=
\exp\!\left(
-\lambda p \int_{\mathbb{R}^2}
\Bigl(1-\frac{1}{1+s\|x\|^{-\alpha}}\Bigr)\,dx
\right)\nonumber\\
&=
\exp\!\left(
-2\pi\lambda p \int_0^\infty \frac{s r^{-\alpha}}{1+s r^{-\alpha}}\,r\,dr
\right).
\end{align}
Evaluating the radial integral for \(\alpha>2\) gives
\[
2\pi \int_0^\infty \frac{s r^{-\alpha}}{1+s r^{-\alpha}}\,r\,dr
=
\pi s^{2/\alpha} C(\alpha),
\]
hence
\[
\mathcal{L}_I(s)=\exp(-\lambda p \pi s^{2/\alpha} C(\alpha)).
\]
Plugging \(s=\theta r^\alpha\) yields \MainEqPsucc. A detailed
derivation appears in \cite{HaenggiJSAC2009}.

\section{Proof of \MainLemmaOne{} and \MainThmTwo}

\subsection*{Proof of \MainLemmaOne}

Since \(V(t{+}1)=x(t)^\top W(t)^\top \Pi W(t)\,x(t)\),
\begin{align*}
\E[V(t{+}1)\mid x(t)]
&=x(t)^\top \E\!\big[W(t)^\top \Pi W(t)\big] x(t)
\\
&\le \lambda_{\max}\!\Big(\E[W(t)^\top \Pi W(t)]\Big)\,x(t)^\top \Pi x(t)
\\
&=\rho\,V(t),
\end{align*}
where we used \(\Pi\succeq 0\) and
\(x^\top A x\le \lambda_{\max}(A)\,x^\top x\) on the subspace
\(\mathbf{1}^\perp\), equivalently through \(x^\top \Pi x\).
This gives \MainEqIdeal.

\subsection*{Proof of \MainThmTwo}

Let \(S(t)\in\{0,1\}\) indicate the event that slot \(t\) produces
a successful averaging exchange under the abstraction of the
main manuscript. By construction of \(q_{\mathrm{lb}}(p)\),
\begin{equation}
\Prb(S(t)=1)\ge q_{\mathrm{lb}}(p).
\end{equation}
When \(S(t)=0\), no effective exchange occurs and thus
\(V(t{+}1)=V(t)\). When \(S(t)=1\), the ideal averaging update
occurs and \MainLemmaOne{} gives
\begin{equation}
\E[V(t{+}1)\mid x(t),S(t)=1]\le (1-\gamma)V(t).
\end{equation}
Taking expectation over \(S(t)\) yields
\begin{align}
\E[V(t{+}1)\mid x(t)]
&\le (1-q_{\mathrm{lb}}(p))\,V(t)+q_{\mathrm{lb}}(p)(1-\gamma)V(t)\nonumber\\
&=\bigl(1-\gamma q_{\mathrm{lb}}(p)\bigr)V(t).
\end{align}
Iterating completes the first inequality in \MainEqContraction.
The bound on the \(\varepsilon\)-consensus time follows by solving
\[
\bigl(1-\gamma q_{\mathrm{lb}}(p)\bigr)^t V(0)\le \varepsilon
\]
and using \(-\log(1-x)\ge x\) for \(x\in(0,1)\).

\section{Proof of \MainPropOne}

Differentiate \(\tilde q(p)\):
\begin{equation}
\tilde q'(p)=a e^{-(a+b)p}-b(1-e^{-ap})e^{-bp}.
\end{equation}
Setting \(\tilde q'(p)=0\) and canceling the common factor
\(e^{-bp}\) gives
\begin{equation}
ae^{-ap}=b(1-e^{-ap}),
\end{equation}
hence
\begin{equation}
e^{-ap}=\frac{b}{a+b},
\qquad
p^\star=\frac{1}{a}\log\!\Big(\frac{a+b}{b}\Big).
\end{equation}
Projection onto \([0,1]\) yields the closed-form rule in
\MainEqPstar. The dense-neighborhood approximation
\(1-e^{-ap}\approx 1\) reduces the main tradeoff to the familiar
throughput form \(p e^{-bp}\), which yields \(p^\star\approx 1/b\).

\section{Additional Remarks on the Main Results}

\begin{remark}[Interpretation]
Define \(K(r)\triangleq \pi r^2\theta^{2/\alpha}C(\alpha)\). Then
\MainEqPsucc{} becomes
\begin{equation}
p_{\mathrm{succ}}(r;p)=\exp(-\lambda p K(r)),
\end{equation}
revealing the exponential penalty in density \(\lambda\), access
probability \(p\), and squared link distance \(r^2\).
\end{remark}

\begin{remark}[Baseline character of \MainThmOne]
\MainThmOne{} is a standard success law for the PPP+Rayleigh
interference-limited baseline. More structured spatial models,
finite-window effects, additive noise, or richer fading laws
may modify the constants and in some cases the functional form.
The theorem should therefore be interpreted as a tractable
baseline reliability law rather than as a claim of universal
exactness for all UAV deployments. The numerical section below
evaluates how far the resulting baseline law remains predictive
once these assumptions are perturbed.
\end{remark}

\begin{remark}[Meaning of \(\gamma\)]
The quantity \(\gamma=1-\rho\) is an \emph{ideal-mixing}
parameter associated with the average-preserving update
abstraction in the absence of wireless failures. It depends
on the geometry, the neighborhood graph, the receiver-selection
rule, and the matching/averaging mechanism, but not on the
SIR success law itself. For a fixed deployment or ensemble,
\(\gamma\) can be estimated numerically, for example by
Monte Carlo averaging of \(W(t)\) or by measuring the empirical
contraction rate of the ideal no-outage process. Better-connected
and more mixing-efficient local update structures typically
yield larger \(\gamma\), whereas sparse or poorly connected
configurations yield smaller \(\gamma\). Thus, \(\gamma\)
captures the purely algorithmic/topological part of the
contraction mechanism, while the wireless layer enters through
the successful-update term below.
\end{remark}

\begin{remark}[Conservativeness of \(q_{\mathrm{lb}}(p)\)]
The bound in \MainEqQlb{} is intentionally conservative: it
uses the half-duplex listening factor, requires the existence
of at least one transmitter inside \(B(0,R)\), and then
evaluates decoding at the worst-case distance \(R\). It
therefore ignores shorter candidate links, the actual
neighbor-distance distribution, and the possibility that
multiple broadcasts could be decodable in the same slot.
Accordingly, \(q_{\mathrm{lb}}(p)\) should be interpreted as a
robust lower bound rather than as a numerically tight expression
for the true successful-update probability. The numerical
section below quantifies this conservativeness through
explicit-interference tightness and proxy-accuracy comparisons.
\end{remark}

\begin{remark}[Interpretation of \MainThmTwo]
The effective contraction factor is the product
\(\gamma\,q_{\mathrm{lb}}(p)\): the system mixes quickly when
(i) the ideal average-preserving update contracts disagreement
strongly (large \(\gamma\)), and
(ii) wireless decoding and medium access produce many
successful updates (large \(q_{\mathrm{lb}}(p)\)). \MainThmTwo{}
is therefore best interpreted as a guaranteed lower-bound
contraction statement under the adopted abstraction, rather than
as an exact prediction of the realized empirical contraction
curve.

If one wishes to account heuristically for unmodeled
coordination overhead in the matching step, the realized
successful-update rate may be viewed as further thinned by
a factor \(\eta_{\mathrm{ctl}}\in(0,1]\), as discussed in
Appendix~A of this supplement. In that case, the same proof
leads to the more conservative replacement
\(q_{\mathrm{lb}}(p)\mapsto \eta_{\mathrm{ctl}}\,q_{\mathrm{lb}}(p)\),
which preserves the same qualitative factorized interpretation
while worsening the constants.
\end{remark}

\begin{remark}[Proxy-optimal versus system-optimal]
The optimizer \(p^\star\) is exact for the simplified proxy
\(\tilde q(p)\) and is used here as an interpretable operating
rule. Its agreement with the empirically best access probability
depends on the quality of the lower-bound approximation and on
how well the chosen distance surrogate represents the links
effectively used by the algorithm. Thus, \(p^\star\) should be
read as a proxy-optimal access rule rather than as a universally
exact system-optimal choice. In particular, the simplified proxy
is most useful for identifying the correct scaling and a plausible
operating region, while the numerical section below compares
its prediction against explicit-interference experiments and
against refined empirical distance averaging.

A particularly natural surrogate is the empirical effective
distance
\[
r_{\mathrm{eff}} \triangleq \sqrt{\E[r^2]},
\]
because the PPP success exponent in \MainEqPsucc{} depends on
\(r^2\). Thus, second-moment distance averaging provides a
first-order way to summarize the realized neighbor-link
distribution when replacing the single worst-case distance
\(R\) by a more representative effective link length.
\end{remark}

\begin{remark}[Swarm scaling]
For fixed \((R,\theta,\alpha)\), \MainEqPstar{} gives
\(p^\star=\Theta(1/\lambda)\). Thus, as the swarm becomes
denser (larger \(\lambda\)), the proxy-optimal randomized
access probability decreases inversely with density, capturing
the fundamental interference--update tradeoff exposed by the
baseline model.
\end{remark}

\section{Full Numerical Validation, Proxy Accuracy, and Robustness Experiments}

This section has three complementary goals. First, we validate
the main interference--access tradeoff through simulations that
use \emph{explicit} instantaneous interference rather than
directly sampling packet success from the closed-form PPP law,
thereby avoiding circular validation of the physical-layer model.
Second, we examine how accurately the proposed proxy rules
predict a useful operating region for the Aloha probability \(p\),
and how conservative the lower-bound quantities in the analysis
are in practice. Third, we assess how robust the resulting
design message remains under changes in receiver abstraction,
additive noise, fading law, and spatial regularity.

Across all experiments, we report the disagreement ratio
\[
\varepsilon(T)\triangleq \frac{V(T)}{V(0)},
\]
where smaller values indicate faster contraction. Unless stated
otherwise, the geometry is kept fixed over the \(T\)-slot horizon,
which should be interpreted as a quasi-static snapshot
approximation appropriate for short control windows when the
displacement per slot is small relative to the communication
radius. In every experiment, the update abstraction remains the
same: each listening node selects one candidate transmitting
neighbor, successful receptions are mapped to a greedy matching,
and each matched pair performs a symmetric averaging step.
This keeps the state evolution average-preserving and makes the
comparison with the analytical framework meaningful.

\subsection{Explicit-interference validation under the PPP baseline}

We first validate the predicted access-probability tradeoff
through an \emph{explicit-interference} simulation that does not
use the closed-form success law of \MainThmOne{} to generate
packet outcomes. Instead, in each slot, the active transmitters
are sampled via slotted Aloha, each listening node selects one
transmitting neighbor within radius \(R\), the desired fading and
all interfering fadings are drawn explicitly, and decoding
succeeds only if the realized instantaneous SIR exceeds the
threshold \(\theta\). Successful receptions are then mapped to a
greedy matching so that the state update remains
average-preserving.

The nominal setup uses \(N=140\) nodes in the toroidal unit
square, communication radius \(R=0.10\), path-loss exponent
\(\alpha=4\), SIR threshold \(\theta=1\), and a horizon of
\(T=600\) slots. A connected finite-window PPP geometry is
generated by sampling i.i.d.\ uniform node locations until the
\(R\)-neighbor graph is connected, and the same geometry and
initialization are used across all values of \(p\) to ensure a fair
comparison. We sweep
\begin{multline}
\label{eq:p_grid_long}
p\in\{0.05,0.08,0.10,0.12,0.14,0.16,0.18,\\
0.20,0.22,0.24,0.26,0.28,0.30,0.35\},
\end{multline}
and report the geometric mean of \(\varepsilon(T)\) together
with \(95\%\) confidence intervals computed in the log domain.

Fig.~\ref{fig:eps_explicit_sir_supp} shows a clear unimodal
dependence of \(\varepsilon(T)\) on \(p\): for small \(p\),
contraction is slow because the medium is underused, whereas
for large \(p\), aggregate interference suppresses successful
receptions and the disagreement ratio increases again. The best
observed operating region is around \(p\approx 0.22\) to
\(0.24\). We also overlay the proxy-based predictions obtained
from the worst-case radius \(R\) and from the empirical effective
distance \(r_{\mathrm{eff}}\). For the realization shown, the
corresponding markers are approximately \(0.112\), \(0.145\),
\(0.184\), and \(0.284\). As expected, the worst-case-radius
proxy is conservative, while the \(r_{\mathrm{eff}}\)-based
closed-form rule moves substantially closer to the empirically
best operating region. Overall, this experiment provides a
direct validation of the main design message: the optimal access
probability is neither too small nor too large, but lies in an
intermediate interference-aware regime.

\begin{figure}[!t]
  \centering
  \includegraphics[width=\linewidth]{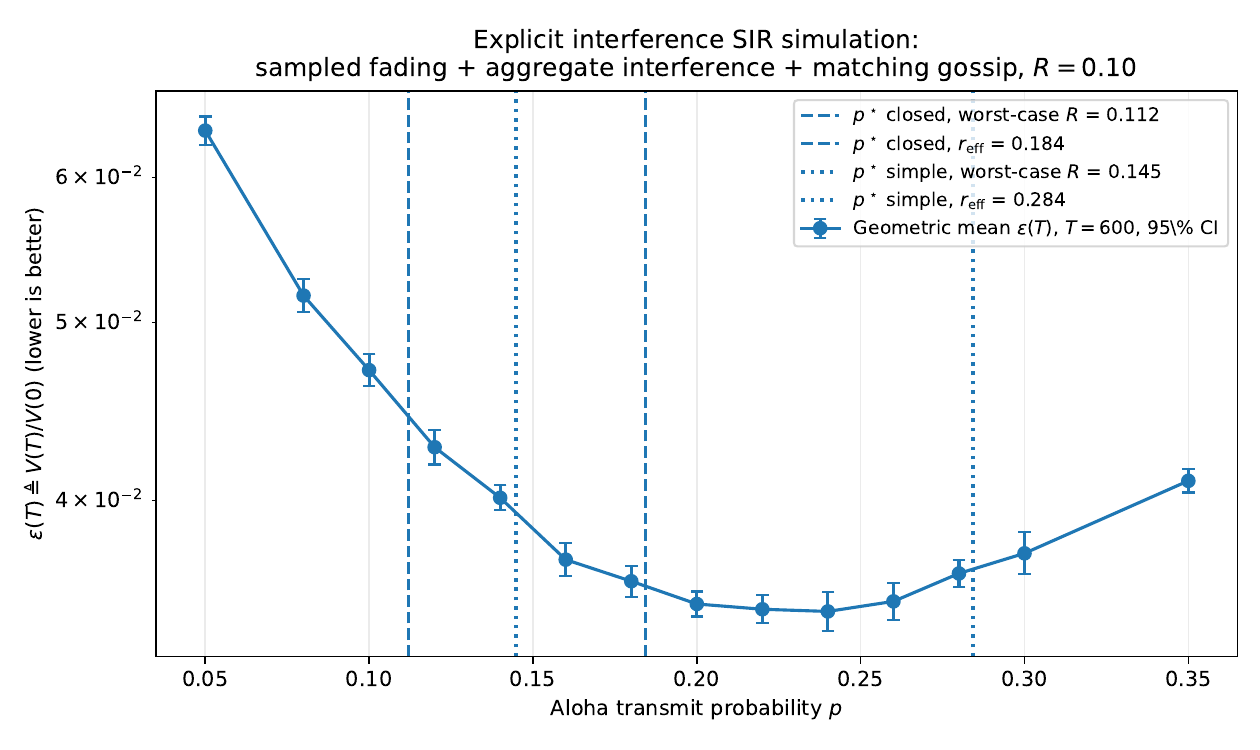}
  \caption{Explicit-interference SIR simulation under the PPP baseline:
  \(\varepsilon(T)=V(T)/V(0)\) versus Aloha transmit probability \(p\), with
  \(95\%\) confidence intervals. The empirical optimum lies in an intermediate
  range of \(p\), while the worst-case-radius proxy remains conservative and
  the \(r_{\mathrm{eff}}\)-based proxy provides a closer prediction.}
  \label{fig:eps_explicit_sir_supp}
\end{figure}

\subsection{Proxy-accuracy comparison: worst-case radius, empirical \(r_{\mathrm{eff}}\), and a refined empirical proxy}

We next quantify how accurately different proxy rules predict
the empirically best Aloha access probability under
explicit-interference simulation. The goal here is to clarify the
modeling gap between the link-level success expression of
\MainThmOne{} and the network-level design rule used to
select \(p\).

The geometry is the same connected finite-window PPP realization
used in the baseline experiment, with \(N=140\),
communication radius \(R=0.10\), path-loss exponent \(\alpha=4\),
SIR threshold \(\theta=1\), and horizon \(T=600\) slots. For
each value of \eqref{eq:p_grid_long}, we run the same
explicit-interference dynamics as in the baseline experiment,
and we again report the geometric mean of \(\varepsilon(T)\)
together with \(95\%\) confidence intervals.

We compare four predictors. The first uses the conservative
worst-case distance substitution \(r=R\) in the closed-form
proxy. The second replaces \(R\) by the empirical effective
distance
\[
r_{\mathrm{eff}} \triangleq \sqrt{\mathbb{E}[r^2]},
\]
computed over the realized neighbor-link distances. This is a
natural second-moment surrogate because the PPP success
exponent depends on \(r^2\). The third uses the
dense-neighborhood simplification of the same rule, evaluated
once with \(R\) and once with \(r_{\mathrm{eff}}\). Finally, we
introduce a refined empirical proxy
\begin{equation}
\label{eq:qbar_clean}
\bar q(p)
=
(1-p)\bigl(1-e^{-\lambda p\pi R^2}\bigr)
\cdot
\frac{1}{|\mathcal{E}|}
\sum_{r\in\mathcal{E}}
\exp\!\Bigl(-\lambda p\pi r^2 \theta^{2/\alpha} C(\alpha)\Bigr),
\end{equation}
where \(\mathcal{E}\) denotes the set of realized
neighbor-link distances in the sampled geometry. This refined
proxy keeps the same availability term as the bound-inspired
rule, but replaces the single worst-case distance by an empirical
average over the actual link distances available in the graph.

Fig.~\ref{fig:eps_proxy_accuracy_supp} shows that all
predictors capture the correct qualitative fact that the best
operating point lies in an intermediate range of \(p\), but their
numerical accuracy differs substantially. The
worst-case-radius predictors are clearly conservative, with
\(p^\star\) values near \(0.112\) and \(0.145\), well below the
empirically best region. Replacing \(R\) by \(r_{\mathrm{eff}}\)
moves the prediction toward the data, but the simple
\(r_{\mathrm{eff}}\)-based approximation overshoots the optimum.
By contrast, the refined empirical proxy selects
\(p^\star\approx 0.160\), which is visibly closer to the
descending part of the empirical curve and substantially reduces
the gap relative to the worst-case rule.

This experiment clarifies the main source of conservativeness
in the analytical design rule. The issue is not the existence of
the interference--access tradeoff itself, which is correctly
captured by all proxies, but rather the use of a single coarse
distance surrogate in place of the actual neighbor-distance
distribution. The refined empirical proxy therefore provides a
useful intermediate interpretation: it preserves the analytical
structure of the availability--reliability decomposition while
showing that a more faithful distance averaging can materially
improve the numerical prediction of a good operating region,
even if it does not exactly recover the empirically best
contraction point.

\begin{figure}[!t]
  \centering
  \includegraphics[width=\linewidth]{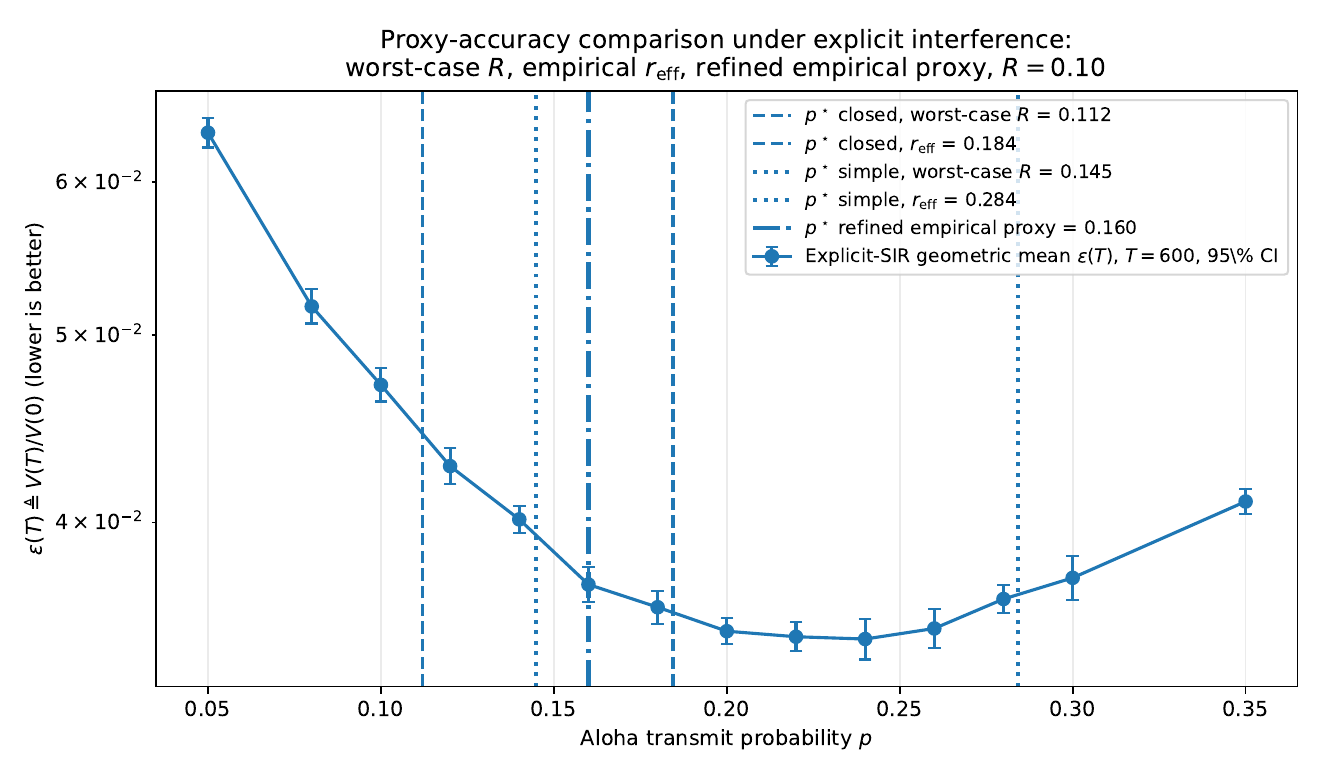}
  \caption{Proxy-accuracy comparison under explicit interference:
  \(\varepsilon(T)=V(T)/V(0)\) versus Aloha transmit probability \(p\), with
  vertical markers corresponding to the worst-case-radius rule, the
  \(r_{\mathrm{eff}}\)-based rule, and the refined empirical proxy. All
  predictors identify the correct intermediate operating regime, while the
  refined empirical proxy gives the closest numerical prediction.}
  \label{fig:eps_proxy_accuracy_supp}
\end{figure}

\subsection{Direct tightness study of the successful-update lower bound}

To quantify the conservativeness of the per-slot
successful-update lower bound \(q_{\mathrm{lb}}(p)\) in
\MainThmTwo, we compare it directly against an empirical
node-level quantity measured from the explicit-interference
simulator. The goal of this experiment is not to claim that
\(q_{\mathrm{lb}}(p)\) is numerically tight, but rather to verify
that it indeed behaves as a conservative lower bound for the
successful-update mechanism induced by the matching-based
abstraction, and to assess how much the prediction improves
when the worst-case distance substitution is replaced by an
empirical distance average.

The setup is the same connected finite-window PPP baseline
used in the main explicit-interference experiment: \(N=140\)
nodes in the toroidal unit square, communication radius
\(R=0.10\), path-loss exponent \(\alpha=4\), SIR threshold
\(\theta=1\), and horizon \(T=600\) slots. For each value of
\eqref{eq:p_grid_long}, we run \(20\) Monte-Carlo trials. In
every slot, transmitters are generated by slotted Aloha, each
listener selects one transmitting neighbor within range \(R\),
fading and aggregate interference are sampled explicitly,
successful decoding is determined by the instantaneous SIR rule,
and successful proposals are filtered through the same greedy
matching used throughout the paper.

The empirical quantity reported in
Fig.~\ref{fig:qlb_tightness_supp} is the \emph{typical-node
matched-receiver probability}, namely the average fraction of
nodes that, in a given slot, are listening nodes participating in
an executed matched successful reception. This is the
appropriate node-level counterpart of the successful-update
probability appearing in the contraction bound. We compare it
against two analytical surrogates: the conservative lower bound
\[
q_{\mathrm{lb}}(p)
=
(1-p)\bigl(1-e^{-\lambda p\pi R^2}\bigr)e^{-\lambda pK(R)},
\]
and the refined empirical proxy \eqref{eq:qbar_clean}, where
\(\mathcal{E}\) is the set of realized neighbor-link distances in
the sampled geometry.

Fig.~\ref{fig:qlb_tightness_supp} shows three clear features.
First, \(q_{\mathrm{lb}}(p)\) remains below the empirical
node-level successful-update probability over the full tested
range of \(p\), confirming that the bound is indeed conservative
for the relevant node-level event. Second, the refined empirical
proxy \(\bar q(p)\) is substantially less conservative and tracks
the empirical curve more closely in both magnitude and shape,
especially in the low-to-moderate \(p\) regime. Third, the
maximizers differ noticeably: the empirical node-level
successful-update rate is largest around
\(p_{\mathrm{emp,upd}}^\star \approx 0.30\), whereas
\(q_{\mathrm{lb}}(p)\) and \(\bar q(p)\) are maximized near
\(0.10\) and \(0.16\), respectively. Thus, the conservative
analytical surrogates correctly identify the qualitative
interference--access tradeoff, but they underpredict the
maximizing \(p\) for this empirical node-level metric. This
node-level maximizer should not be conflated with the
empirically fastest-mixing \(p\) for the disagreement metric
\(\varepsilon(T)\); rather, it isolates one particular ingredient
of the contraction mechanism.

This experiment therefore refines the interpretation of the
theory. The role of \(q_{\mathrm{lb}}(p)\) is to provide a
guaranteed lower-bound contribution to the contraction
mechanism, not to act as an exact numerical fit to the realized
successful-update probability. At the same time, replacing the
worst-case distance \(R\) by an empirical average over the
actual neighbor-link lengths materially narrows the gap, which
helps explain why refined empirical proxies produce more
accurate operating-region predictions than the purely
worst-case rule.

\begin{figure}[!t]
  \centering
  \includegraphics[width=\linewidth]{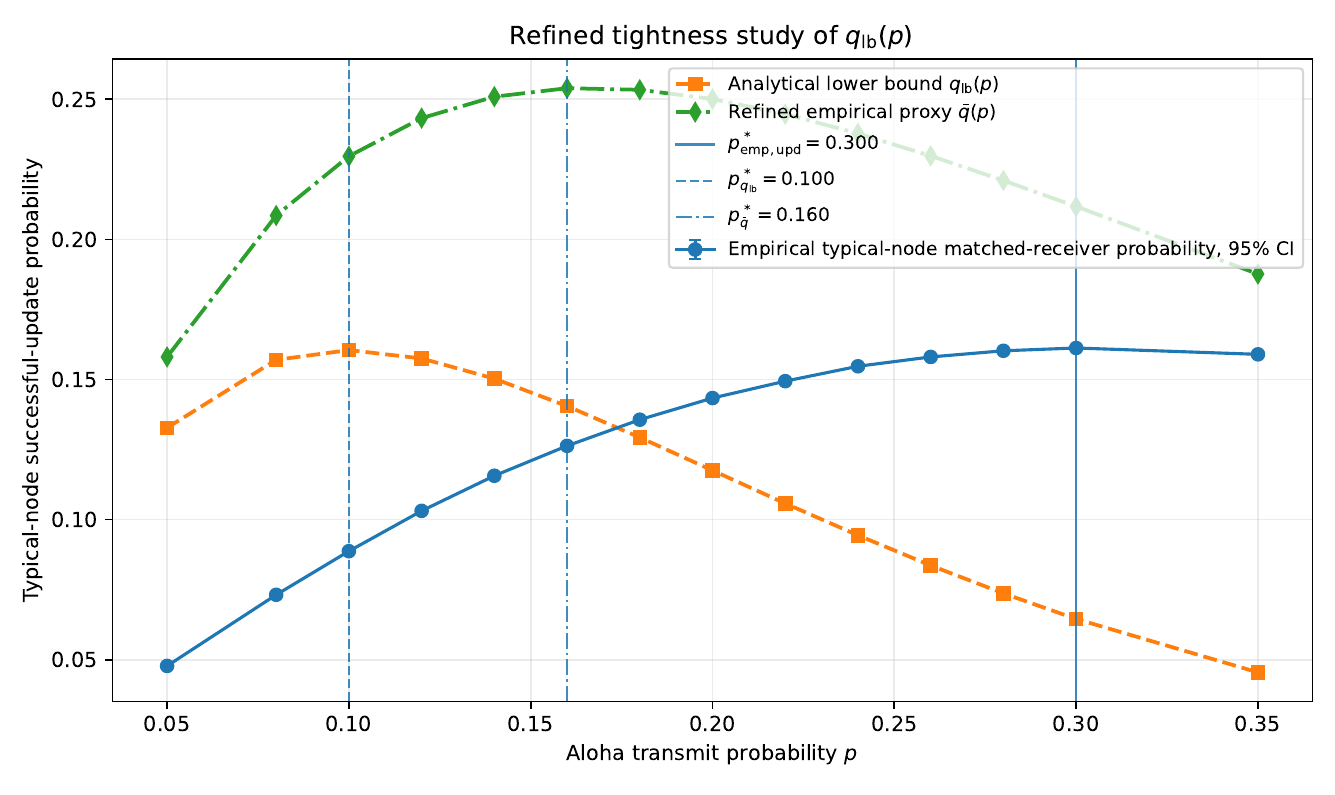}
  \caption{Refined tightness study of \(q_{\mathrm{lb}}(p)\):
  empirical typical-node matched-receiver probability versus Aloha transmit
  probability \(p\), compared with the analytical lower bound
  \(q_{\mathrm{lb}}(p)\) and the refined empirical proxy \(\bar q(p)\).
  The analytical lower bound is conservative throughout, while the refined
  empirical proxy tracks the empirical trend more closely.}
  \label{fig:qlb_tightness_supp}
\end{figure}

\subsection{Empirical interpretation of the ideal-mixing parameter \texorpdfstring{$\gamma$}{gamma}}

To make the role of the ideal-mixing parameter \(\gamma\) in
\MainThmTwo{} more concrete, we consider an \emph{ideal
no-outage} experiment in which the wireless decoding layer is
removed entirely and every proposed exchange succeeds. The
purpose of this experiment is to isolate the
\emph{algorithmic/topological} part of the contraction mechanism
from the wireless-thinning effect. In particular, the resulting
estimate should be interpreted as an empirical counterpart of
the ideal-mixing quantity \(\gamma=1-\rho\) introduced in
\MainLemmaOne.

The update mechanism remains otherwise unchanged. In each
slot, nodes are split into transmitters and listeners through
slotted Aloha with probability \(p\), each listener selects one
transmitting neighbor according to the same single-neighbor
receiver abstraction used in the baseline model, and the
resulting proposals are mapped to a greedy matching. The only
difference with respect to the full wireless simulation is that
every proposed matched exchange is executed, so that no
SIR-related failures occur. Hence, this experiment keeps the
same random matching structure induced by the MAC and
receiver rule, while removing the wireless success/failure layer.

We compare two geometry ensembles: a PPP ensemble and a
perturbed-lattice ensemble. To make the comparison as
controlled as possible, each realization is required to be
connected and approximately degree-matched to a target mean
degree \(5.5\), with tolerance \(\pm 0.2\). We use \(N=140\),
\(L=1\), a horizon \(T=350\), \(10\) independent geometry
realizations per ensemble, and \(6\) Monte-Carlo runs per
geometry. For each operating point \(p\), we estimate the
one-step ideal contraction ratio through
\begin{equation}
\hat{\rho}(p)
\triangleq
\frac{1}{T}\sum_{t=0}^{T-1}\frac{V(t+1)}{V(t)},
\qquad
\hat{\gamma}(p)\triangleq 1-\hat{\rho}(p).
\end{equation}

Fig.~\ref{fig:gamma_ideal_no_outage_supp} shows that the
perturbed-lattice ensemble has a consistently larger
\(\hat{\gamma}(p)\) than the PPP ensemble across the full tested
range of \(p\), even though the two ensembles are essentially
matched in mean degree. This is informative in two ways. First,
it confirms that \(\gamma\) is not merely a formal parameter in
the theorem, but a numerically observable quantity capturing
the strength of the ideal average-preserving mixing mechanism.
Second, it shows that more regular spatial structure can improve
the ideal-mixing part of the dynamics, independently of
wireless decoding failures.

It is also worth emphasizing that \(\hat{\gamma}(p)\) remains
\(p\)-dependent in this ideal experiment. This does not
contradict the interpretation of \(\gamma\) as an algorithmic
quantity: although wireless outages have been removed, the
Aloha transmit/listen split still changes the random update
structure and hence the effective matching process. In other
words, the present experiment isolates the \emph{non-wireless}
part of the contraction while still retaining the MAC-induced
randomness of the average-preserving broadcast-gossip
abstraction.

Overall, this experiment complements the explicit-interference
results of the preceding subsections. The explicit-SIR
simulations show how wireless failures shape the overall
operating point, whereas the present no-outage study clarifies
how the ideal-mixing term itself varies with geometry and
access probability. This supports the factorized interpretation of
\MainThmTwo, in which the realized contraction depends jointly
on an ideal-mixing component \(\gamma\) and a wireless
successful-update component.

\begin{figure}[!t]
  \centering
  \includegraphics[width=\linewidth]{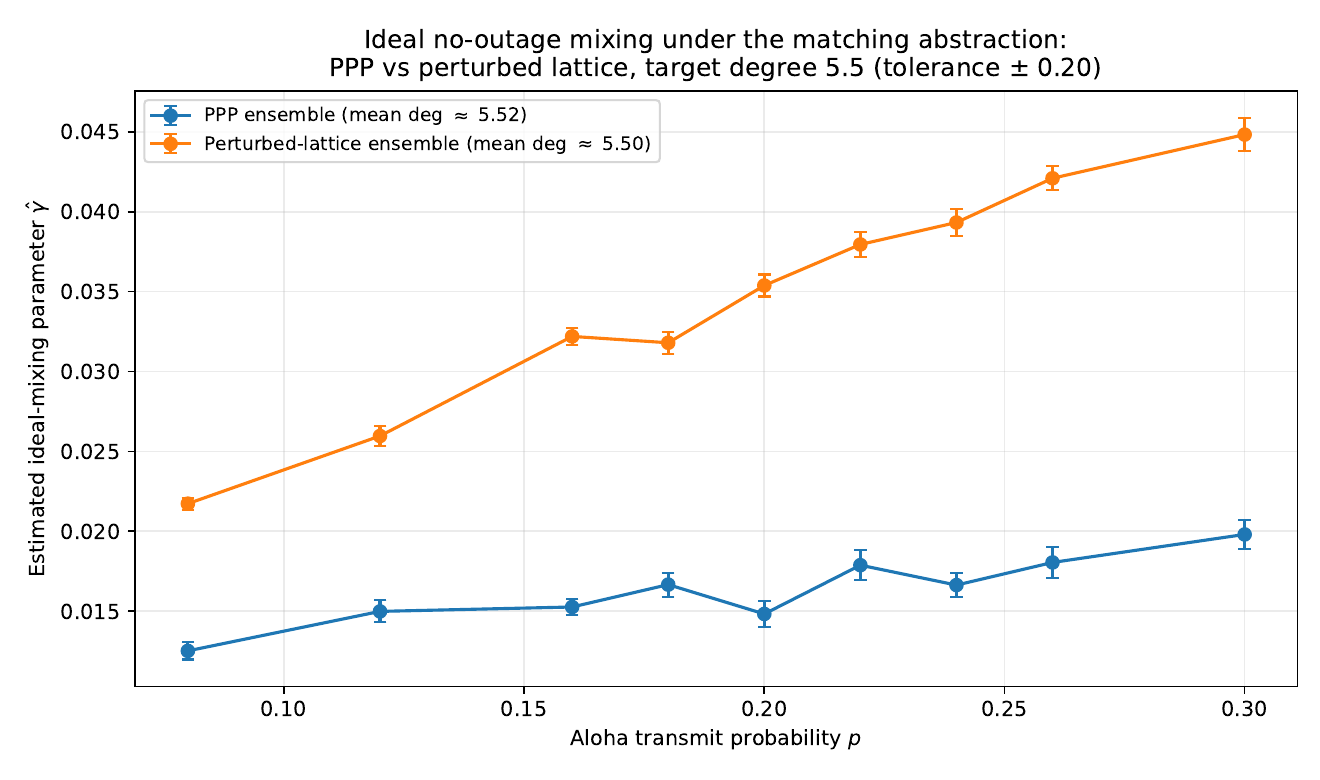}
  \caption{Ideal no-outage estimate of the mixing parameter:
  \(\hat{\gamma}(p)\) for PPP and perturbed-lattice ensembles under the same
  matching-based update abstraction, with connected and approximately
  degree-matched realizations. The more regular perturbed-lattice geometry
  yields a consistently larger ideal-mixing parameter, indicating that spatial
  structure affects the algorithmic/topological part of the contraction even
  before wireless outages are introduced.}
  \label{fig:gamma_ideal_no_outage_supp}
\end{figure}

\subsection{Receiver-abstraction robustness under explicit interference}

The analytical model adopts a conservative one-link receiver
abstraction in which each listening node selects a single
transmitting neighbor and attempts to decode only that link. To
assess whether the main design conclusion depends strongly on
the specific choice of receiver rule, we compare two variants
under the same explicit-interference simulation:
(i) the baseline rule in which a listener selects one active
transmitting neighbor uniformly at random, and
(ii) a more structured rule in which the listener selects the
nearest active transmitter.

The geometry and physical-layer assumptions are the same as in
the explicit-SIR baseline: a connected finite-window PPP
realization with \(N=140\) nodes in the toroidal unit square,
communication radius \(R=0.10\), path-loss exponent
\(\alpha=4\), SIR threshold \(\theta=1\), and horizon \(T=600\)
slots. For each value of \eqref{eq:p_grid_long}, we perform
\(20\) Monte-Carlo runs and report the geometric mean of
\[
\varepsilon(T)\triangleq \frac{V(T)}{V(0)}
\]
together with \(95\%\) confidence intervals computed in the
log domain.

Fig.~\ref{fig:eps_receiver_rule_robustness_supp} shows that
both receiver rules produce the same qualitative behavior: the
disagreement ratio is clearly unimodal in \(p\), and the best
operating region remains intermediate rather than shifting to
either extreme. In both cases, the best observed operating point
occurs at \(p\approx 0.24\). Quantitatively, the uniform-choice
rule achieves a slightly smaller minimum disagreement ratio,
\[
\varepsilon(T)\approx 0.0348
\quad \text{at } p=0.24,
\]
whereas the nearest-active-transmitter rule yields
\[
\varepsilon(T)\approx 0.0367
\quad \text{at } p=0.24.
\]
Thus, changing the receiver rule modifies the constants but
does not alter the central access-design conclusion.

The fact that the nearest-transmitter rule is not superior in
this setting is also informative. Although choosing the nearest
active transmitter improves the quality of the selected wireless
link, it can also concentrate many listeners onto the same
nearby transmitters. Under the matching abstraction, each
transmitter can participate in at most one executed averaging
exchange per slot, so this concentration may reduce the
effective utilization of successful proposals. By contrast, the
uniform-choice rule spreads proposals more evenly across active
transmitters, which can improve the final matching yield even
if some selected links are, on average, weaker. This robustness
check therefore does not validate the one-link abstraction as
physically exact, but it does show that the qualitative
access-design conclusion is not fragile to a more structured
receiver-selection rule.

\begin{figure}[!t]
  \centering
  \includegraphics[width=\linewidth]{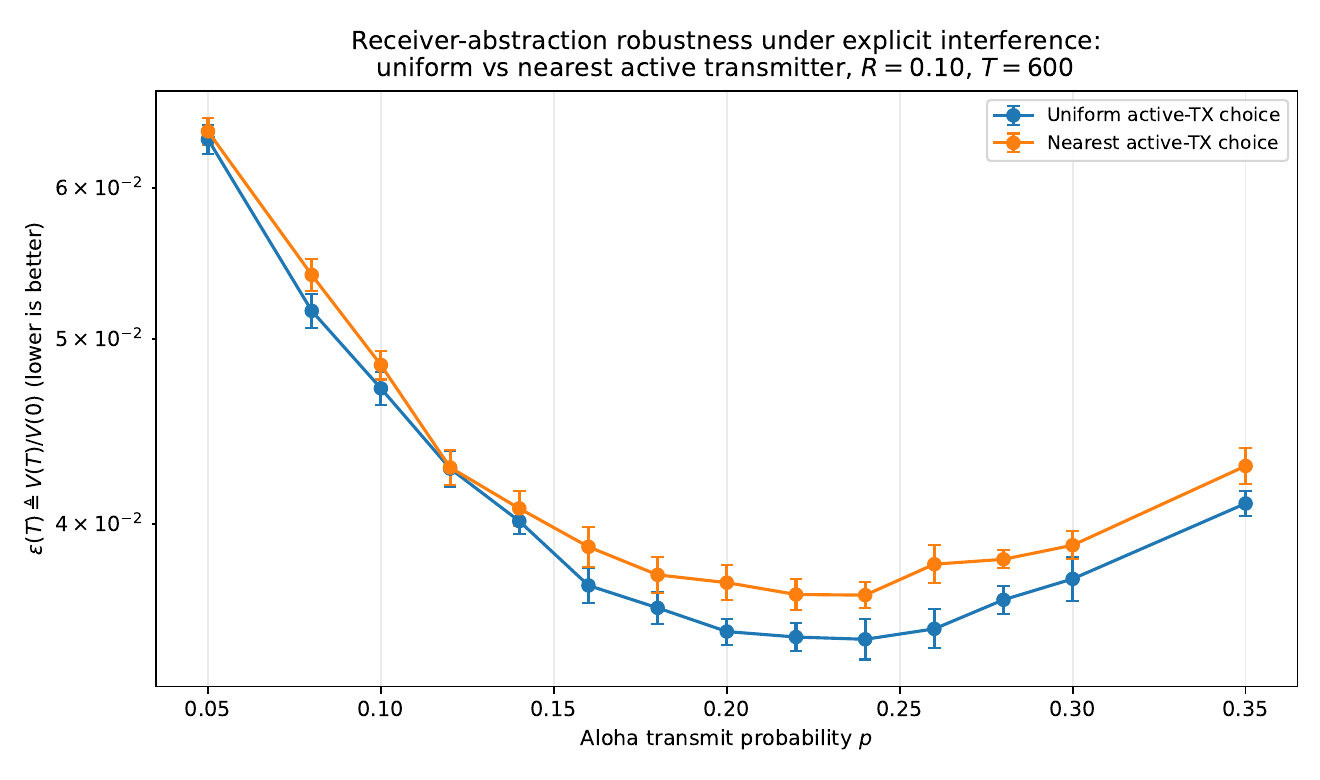}
  \caption{Receiver-abstraction robustness under explicit interference:
  comparison between uniform active-transmitter selection and nearest
  active-transmitter selection. Both receiver rules exhibit the same unimodal
  dependence on \(p\) and the same intermediate best operating region, showing
  that the access-design message is robust to this change in receiver
  abstraction.}
  \label{fig:eps_receiver_rule_robustness_supp}
\end{figure}

\subsection{SINR robustness: explicit noise sensitivity beyond the interference-limited baseline}

To examine whether the main design conclusion depends
critically on the interference-limited approximation, we repeat
the explicit physical-layer simulation with \emph{additive noise}
and compare the resulting signal-to-interference-plus-noise
ratio (SINR) dynamics against the original SIR baseline. This
experiment is not intended as a complete communication-link
parameterization with carrier frequency, bandwidth, receiver
noise figure, and antenna modeling. Rather, it serves as a
normalized robustness check showing how the contraction
behavior changes when a nonzero noise floor is introduced into
the decoding rule.

The geometry and update abstraction are kept identical to the
explicit-SIR baseline: \(N=140\) nodes are placed in a connected
finite-window PPP realization on the toroidal unit square, the
communication radius is \(R=0.10\), the path-loss exponent is
\(\alpha=4\), the decoding threshold is \(\theta=1\), and the
horizon is \(T=600\) slots. In each slot, transmitters are
selected independently with probability \(p\), each listening
node chooses one transmitting neighbor within range \(R\),
desired and interfering fading coefficients are sampled
explicitly, and successful receptions are mapped to a greedy
matching so that the state update remains average-preserving.

The only difference relative to the SIR experiment is that packet
reception is now decided by the explicit SINR rule
\begin{equation}
\mathrm{SINR}_{j\to i}
=
\frac{h_{ji}\,r_{ij}^{-\alpha}}
{\sum_{k\in\mathcal{T}\setminus\{j\}} h_{ki}\,r_{ki}^{-\alpha} + \sigma^2},
\end{equation}
with success declared whenever \(\mathrm{SINR}_{j\to i}>\theta\).
We compare three normalized noise levels,
\[
\sigma^2 \in \{0,\;10^{-3},\;10^{-2}\},
\]
where \(\sigma^2=0\) corresponds to the original SIR baseline.
As before, we report the geometric mean of \(\varepsilon(T)\)
with \(95\%\) confidence intervals computed in the log domain.

Fig.~\ref{fig:eps_sinr_robustness_supp} shows that the three
curves remain very close over the full range of \(p\). In
particular, the performance remains clearly unimodal, and the
best operating region stays in the same intermediate range as
in the interference-limited case. Mild additive noise therefore
does not qualitatively alter the access-probability tradeoff:
when \(p\) is too small, the system suffers from insufficient
update opportunities, whereas when \(p\) is too large, the
combined effect of interference and noise again slows
contraction. The small separation between the three curves
indicates that, for the present normalized regime, the dominant
mechanism remains interference rather than receiver noise.

This experiment should be interpreted as a robustness check
rather than as a claim that the system is universally
interference-limited in all UAV settings. Its main message is
that the proposed MAC-level tuning rule remains stable when
the decoding model is perturbed from SIR to SINR by a
reasonable nonzero noise term. Thus, the preference for an
intermediate access probability is not an artifact of setting the
thermal-noise term exactly to zero.

\begin{figure}[!t]
  \centering
  \includegraphics[width=\linewidth]{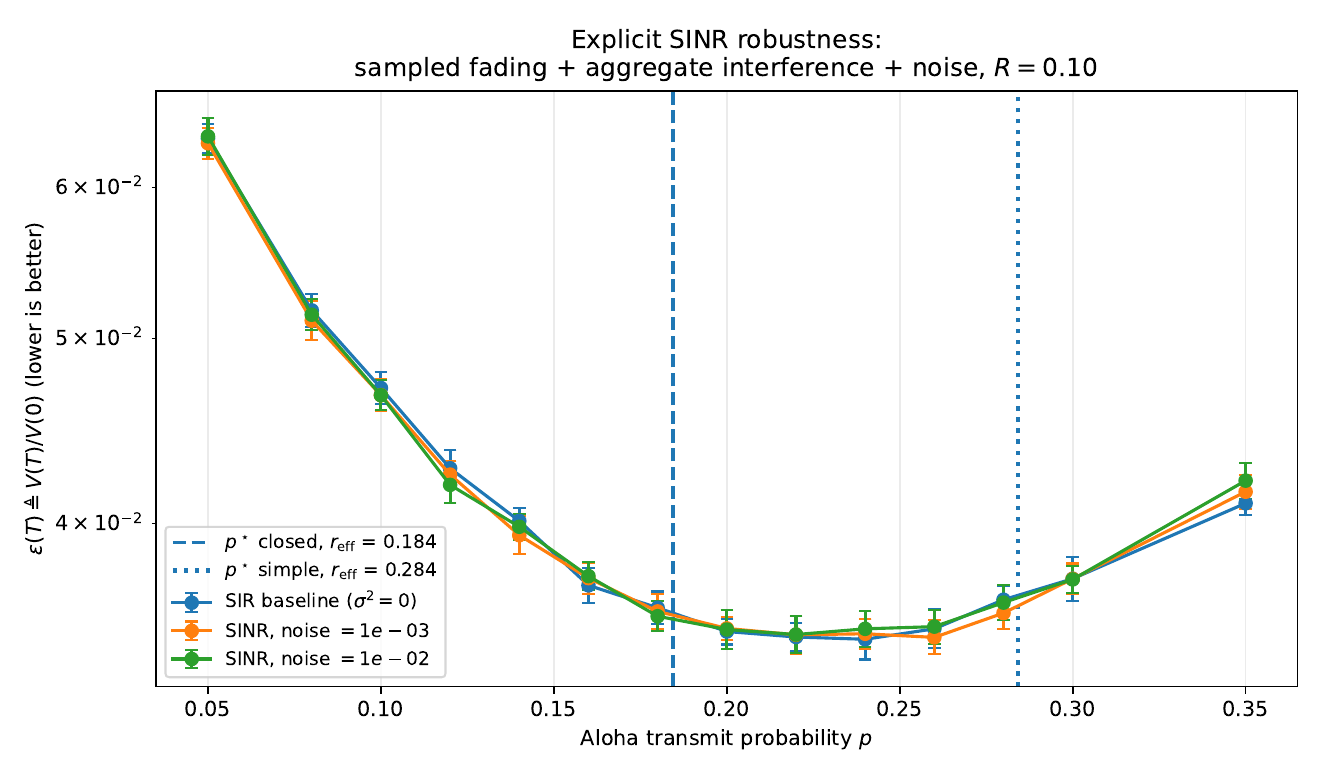}
  \caption{Explicit SINR robustness: \(\varepsilon(T)=V(T)/V(0)\) versus
  Aloha transmit probability \(p\) under three noise levels,
  including the SIR baseline \((\sigma^2=0)\). The favorable operating region
  remains essentially unchanged, indicating that the design trend is robust to
  mild additive noise.}
  \label{fig:eps_sinr_robustness_supp}
\end{figure}

\subsection{Fading robustness: explicit interference beyond the Rayleigh baseline}

The analytical development adopts Rayleigh fading for
tractability, but UAV links may exhibit milder small-scale
fluctuations than a pure Rayleigh model. To test whether the
main access-probability design trend depends strongly on this
assumption, we repeat the explicit-interference simulation under
several fading laws while keeping the geometry, MAC rule, and
averaging abstraction unchanged.

The nominal setup is the same as in the explicit-SIR baseline:
\(N=140\) nodes are placed in a connected finite-window PPP
realization on the toroidal unit square, the communication
radius is \(R=0.10\), the path-loss exponent is \(\alpha=4\),
the SIR threshold is \(\theta=1\), and the horizon is \(T=600\)
slots. In each slot, transmitters are selected independently with
probability \(p\), each listening node chooses one transmitting
neighbor within range \(R\), the desired and interfering channel
gains are sampled explicitly, and successful receptions are
mapped to a greedy matching so that the update remains
average-preserving.

We compare three fading models with unit mean power gain:
Rayleigh fading, corresponding to an exponential gain
distribution, and two Nakagami-\(m\) alternatives with \(m=2\)
and \(m=3\). The latter provide progressively less severe fading
fluctuations and may be viewed as a simple robustness check
against more regular channel conditions. As in the previous
experiments, we report \(\varepsilon(T)\) through its geometric
mean over Monte-Carlo trials together with \(95\%\) confidence
intervals computed in the log domain.

Fig.~\ref{fig:eps_fading_robustness_supp} shows that the
dependence of \(\varepsilon(T)\) on \(p\) remains clearly
unimodal for all three fading models. The empirically best
operating region stays in the same intermediate range,
approximately \(p\approx 0.20\) to \(0.24\), and the differences
between Rayleigh and Nakagami-\(m\) are moderate over the
full sweep. Around the best operating region, the Nakagami
cases yield slightly smaller disagreement ratios, while at larger
access probabilities the curves remain close and preserve the
same qualitative degradation trend. Thus, the preference for an
intermediate Aloha probability is not tied to the Rayleigh
assumption alone.

This experiment should be interpreted as a sensitivity check
rather than as a complete air-to-air channel model. Its role is
to show that the MAC-level tuning message of the paper is
robust to replacing the baseline Rayleigh law by milder fading
distributions. In particular, the proposed operating rule
continues to identify the correct design region even when the
small-scale fading statistics are perturbed away from the
nominal analytical model.

\begin{figure}[!t]
  \centering
  \includegraphics[width=\linewidth]{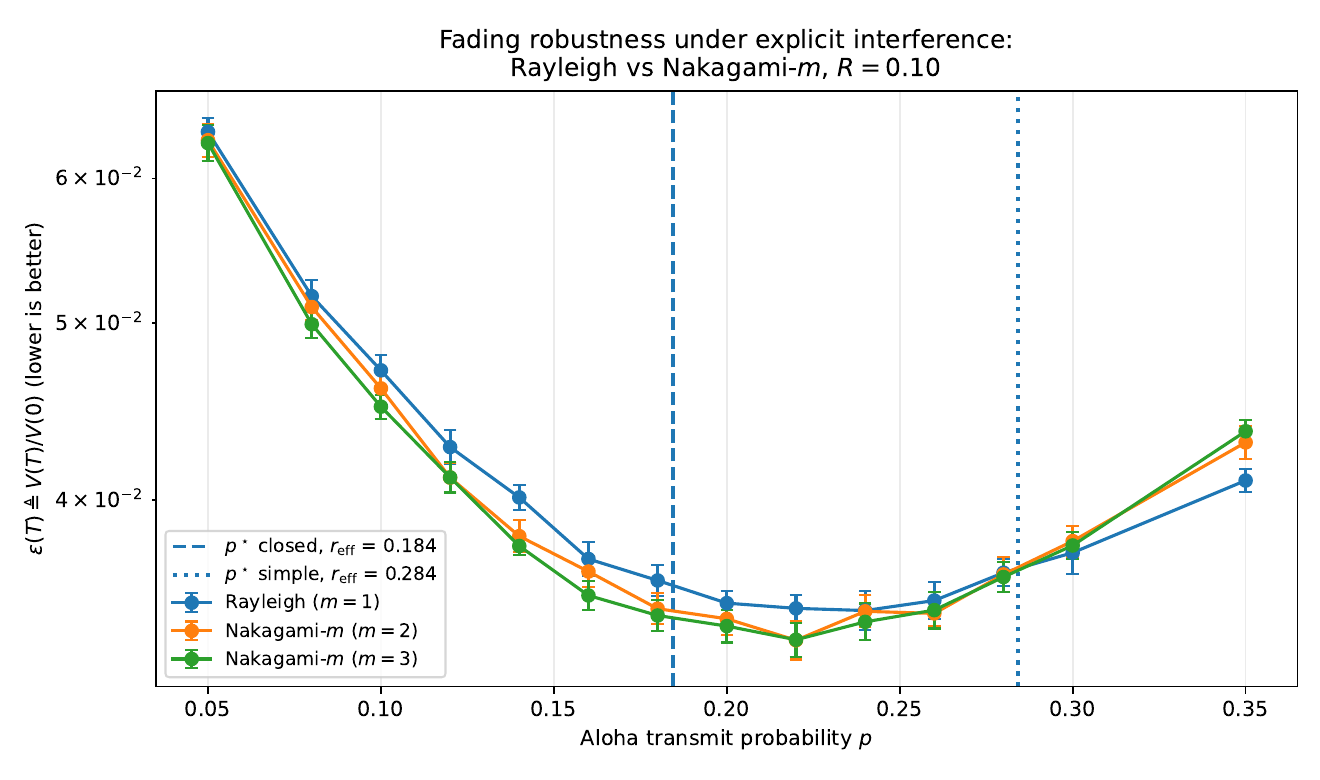}
  \caption{Fading robustness under explicit interference:
  \(\varepsilon(T)=V(T)/V(0)\) versus Aloha transmit probability \(p\) for
  Rayleigh and Nakagami-\(m\) fading models. The favorable operating region
  remains essentially unchanged, indicating that the main design trend is
  robust to moderate variations in the fading law.}
  \label{fig:eps_fading_robustness_supp}
\end{figure}

\subsection{Spatial robustness: PPP versus perturbed-lattice ensembles under matched connectivity constraints}

To assess how strongly the observed access-probability tradeoff
depends on the PPP placement model, we compare two
\emph{geometry ensembles} under the same explicit-interference
simulation rule: a PPP ensemble and a perturbed-lattice
ensemble. The objective is not to claim that these two models
exhaust all relevant swarm formations, but rather to test
whether the qualitative design trend in \(p\) persists when the
node layout is made substantially more regular.

For each geometry realization, we first construct an
\(R\)-neighbor graph by selecting the smallest radius that
ensures connectivity while also targeting a prescribed mean
degree. This produces connected graphs with broadly
comparable neighborhood density across the two ensembles,
thereby reducing the extent to which the comparison is driven
by trivial sparsity effects. The perturbed lattice is obtained by
jittering a regular grid, while the PPP baseline is generated by
i.i.d.\ uniform node placements. Because the connectivity
constraint is enforced realization by realization, the degree
matching is only approximate in finite samples. The comparison
should therefore be interpreted as a robustness check under
approximately matched local density rather than as a perfectly
degree-controlled experiment.

The explicit physical-layer simulation then proceeds exactly as
in the SIR baseline: at each slot, nodes transmit independently
with probability \(p\), each listener chooses one transmitting
neighbor, Rayleigh fading is sampled explicitly on the desired
and interfering links, decoding succeeds if the instantaneous
SIR exceeds \(\theta\), and successful receptions are mapped to a
greedy matching so that the update remains average-preserving.
We use \(N=140\), \(L=1\), \(\alpha=4\), \(\theta=1\), a horizon
\(T=350\), and evaluate
\[
p\in\{0.08,0.12,0.16,0.18,0.20,0.22,0.24,0.26,0.30\}.
\]
For each ensemble, results are aggregated over \(10\)
independent geometry realizations and \(6\) Monte-Carlo runs
per geometry.

Fig.~\ref{fig:eps_spatial_robustness_supp} shows that both
ensembles still favor an intermediate Aloha probability, so the
non-monotone dependence on \(p\) is not specific to the PPP
baseline. At the same time, the absolute performance level
differs markedly across the two spatial models: the
perturbed-lattice ensemble achieves substantially smaller
disagreement ratios than the PPP ensemble over the entire range
of \(p\). This indicates that spatial regularity can materially
improve wireless averaging performance by reducing
unfavorable geometry fluctuations, even when connectivity and
local density are only approximately matched. Accordingly, the
observed performance gap should be read primarily as a
robustness indicator showing that the constants are
geometry-sensitive, not as a claim of exact degree-controlled
superiority.

This experiment therefore supports a nuanced interpretation of
the theory. The PPP model remains useful as a tractable
irregular-deployment baseline and correctly captures the
existence of an interference--access tradeoff, but the constants
governing the achievable contraction can change significantly
under more structured formations. In particular, the preferred
\emph{region} of \(p\) is stable, whereas the absolute
contraction level is geometry-sensitive. That distinction is
important for interpreting the closed-form rule as a baseline
design heuristic rather than a universal exact optimizer across
all swarm layouts.

\begin{figure}[!t]
  \centering
  \includegraphics[width=\linewidth]{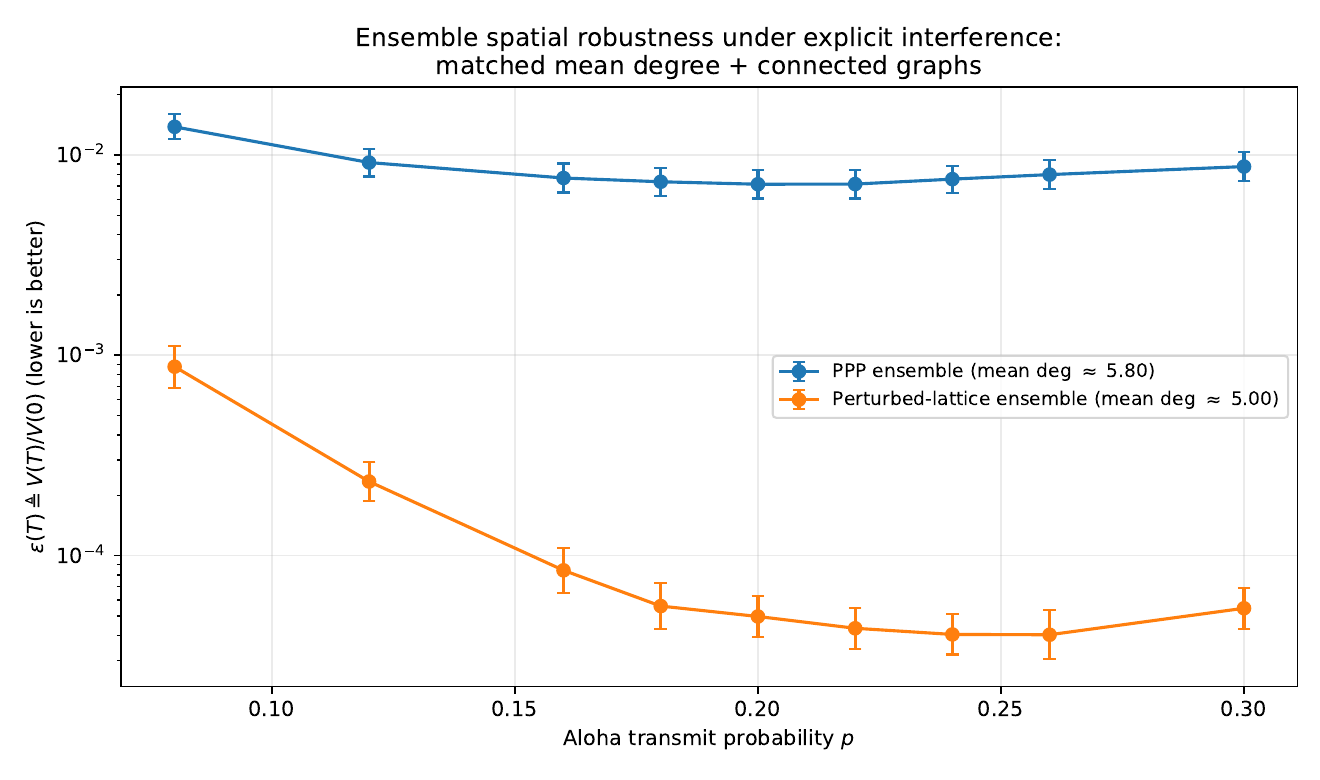}
  \caption{Ensemble spatial robustness under explicit interference:
  PPP and perturbed-lattice geometries with approximately matched connectivity
  and mean degree. Both ensembles retain an intermediate optimal region in
  \(p\), but the more regular perturbed-lattice geometry yields substantially
  smaller disagreement ratios, showing that spatial regularity affects the
  constants even when the qualitative tradeoff persists.}
  \label{fig:eps_spatial_robustness_supp}
\end{figure}

\subsection{Summary of numerical findings}

Taken together, the experiments support a coherent
interpretation of the analytical framework. The
explicit-interference baseline confirms the predicted
non-monotone dependence on the access probability \(p\). The
proxy and tightness studies show that the worst-case
lower-bound quantities are conservative, while empirical
distance averaging improves numerical prediction. The ideal
no-outage experiment clarifies the role of the ideal-mixing
parameter \(\gamma\), and the receiver-rule, SINR, fading, and
spatial robustness studies show that the central access-design
message is stable under several meaningful deviations from the
nominal baseline. The main conclusion of this section is
therefore not that the analytical surrogates are numerically
exact, but that they correctly capture the underlying
interference--mixing tradeoff and provide a tractable guideline
for identifying an intermediate Aloha operating region.

\end{document}